\documentclass[11pt,a4paper,leqno]{article}

\usepackage[english]{babel}
\usepackage{latexsym}
\usepackage{amsmath, amssymb}
\usepackage{graphicx}
\usepackage{epsfig}
\usepackage{dsfont}
\usepackage[applemac]{inputenc} 
\usepackage[T1]{fontenc}      

\usepackage{latexsym,dsfont,textcomp,amsmath}
\usepackage{color}
\usepackage{stmaryrd}
\usepackage{fancyhdr}
\usepackage{float}

\newtheorem{teorema}{Theorem}[section]                                                                                                                                                                                     
\newtheorem{lemmas}[teorema]{Lemma}
\newtheorem{proposizione}[teorema]{Proposition}
\newtheorem{corollario}[teorema]{Corollary}

\newtheorem{remarque}{Remark}[section]
\newtheorem{definizione}{Definition}[section]
\newtheorem{assumption}{Assumption}

\newtheorem{principle}{Principle}

\newenvironment{proof}{{\textbf Proof:} }{\smallskip \begin{flushright}
$\Box$
\end{flushright}}

\title{Asset pricing under uncertainty}

\author{Simone SCOTTI\footnote{Laboratoire de Probabilit\'{e}s et Mod\`{e}les Al\'{e}atoires, Universit\'{e} Paris 7,
CNRS UMR7599, email:  simone.scotti@univ-paris-diderot.fr } }

\begin{document}

\maketitle

\begin{abstract}
  We study the effect of parameter uncertainty on a stochastic diffusion 
  model, in particular the impact on the pricing of contingent claims, using methods from the 
  theory of Dirichlet forms.
  We apply these techniques to hedging procedures in order to compute  the sensitivity  of 
  SDE trajectories with respect to parameter perturbations.  
  We show that this analysis can justify endogenously the presence of a bid-ask spread on
  the option prices.
  We also prove that if the stochastic differential equation admits a closed form representation 
  then the sensitivities have closed form representations.
  
  We examine the case of log-normal diffusion and we show that this framework leads to
   a smiled implied volatility surface coherent with historical data. 

\vspace{0.5cm}

\vspace{0.5cm}
 { \bf AMS:} 60H30, 91B16 and 91B70


\end{abstract}

\section{Introduction}

The purpose of this article is to study the effect of parameter uncertainty on asset pricing.
If we consider, for instance, the Black and Scholes model, the price of a stock is modeled as a 
log-normal process. This diffusion process depends on a random source, represented by a 
Brownian motion, that models the market uncertainty and the unknown 
future evolution of the stock price. This diffusion depends also on two parameters, usually known
as the drift and the volatility. The latter plays a crucial role in asset pricing. 
As such, a second source of uncertainty appears since the values of these parameters 
have to be estimated by financial agents. In practice, in order to sell an option, one has to fix a 
price using his pricing method, which depends on the diffusion 
parameters. Since these diffusion parameter are unknown, one has to plug into his pricing model their 
estimated values.
However, these estimated values are uncertain, in other words, they are random variables. 
A natural question arising is therefore the impact of this second source of uncertainty.

A large literature exists on this subject. We may in particular refer to the uncertain volatility 
model (UVM) which 
take into account the difficulties of calibrating the volatility in the Black and Scholes model,  
see Avellaneda et al. \cite{bib:Avelaneda} and  Lyons \cite{bib:Lyons}.
In UVM, the parameter,  i.e. the volatility, is an unknown parameter which belongs to a interval. 
However, this kind of approach neglects the way practitioners implement their pricing method 
by using estimated parameters. On the other hand, practitioners often overlook
the inherent uncertainty in their estimated parameters. 

In our model, we assume that the stock price follows a stochastic differential equation (SDE) with 
fixed but unknown parameters. As a consequence, only one random source exists in the market,
 which is therefore complete. 
However, the financial agents, assumed to be option sellers in our analysis,
know that the right model is an SDE but they do not known the true values of the parameters.
Therefore, they estimate the parameters and use these estimates to price the options. 
In other words, the parameters are replaced by statistical estimators of the true 
parameter values, that is, they are random variables  characterized by a variance and a 
bias. 
As a consequence, the option seller uses a different SDE, i.e. an SDE with 
different parameters to compute the option prices. It follows that the probability 
space used in this paper must describe two random sources: the market uncertainty, 
i.e. the Brownian motion driving the SDE, and the parameter uncertainty, 
that affects the calculations made by the option seller. 
However, the two uncertainties do not play the same role. 
Since the law of the market uncertainty is well defined while the law of the parameter estimators is poorly known. 
Because of this, we assume known only the biases and the variances of these estimators.  
Our goal is to study how the bias and the variance of the parameter estimators  affect the 
basics of the asset prices and the related option prices.

An important difference between our work and the UVM model is that we consider a large class 
of underlying diffusions and not uniquely the log-normal one. We assume that  the true underlying 
diffusion is unique and follows a stochastic differential equation with continuous paths,
\begin{equation*}
dS_t = S_t \, \mu \, dt + S_t \, \sigma(t, \, S_t) \, dB_t \, .
\end{equation*}
This diffusion can depend on several parameters, and 
not only on the Black-Scholes volatility, that are 
to be estimated by the option seller. 
We assume that the uncertainty is on the volatility and  carried by the 
parameters used to describe this volatility.  
Our method leads to a option pricing calculus that  enables the option seller to control the  
risk due to parameter uncertainty  assuming an acceptable residual risk that can be 
evaluated using a Value-at-Risk technique,
we propose a principle to price an option and we find a selling price such that the probability 
of loss of the option seller is smaller that a given tolerance $\alpha$, we explain our pricing principle 
in Section 3.2.

In the UVM case, the strategy is to apply stochastic control techniques to super-hedge 
a contingent claim. 
The principal drawback of the super-replication technique is that the super-hedging 
cost is too high and the corresponding  strategy  too conservative. 
As an example, Kramkov \cite{bib:Kramkov} shows that the super-hedging cost of a 
call corresponds 
to the value of the option under the least favorable martingale measure. Likewise, Bellamy 
and Jeanblanc \cite{bib:Bellamy} prove that in a jump-diffusion model the price range for a 
call option  corresponds to the interval given by the no-arbitrage conditions that is the 
super-hedging strategy of a call option is to buy the underlying. 
It seems therefore that super-hedging is not an effective methodology since the option
buyer will prefer to buy the underlying rather than pay the same price to have an option. 
Many attempts to overcome this problem have been studied in literature. Avellaneda et al. 
\cite{bib:Avelaneda} constrains the volatility between two levels. Another approach is the pricing
 via utility maximization, see for instance El Karoui and Rouge \cite{bib:ElKaroui}, or via 
 local risk minimization, see for instance Follmer and Schweizer \cite{bib:Follmer}.
We also recall the paper of Denis and Martini \cite{bib:Denis-1}, where a more 
general framework is introduced in order to study super-replication and model uncertainty.
Their methodology makes use of  Choquet capacities and is related 
to our approach. Recently, a generalization of this approach is proposed by Soner et al  
\cite{bib:Soner}.

In our analysis, we assume that the option seller has to take some risk in order 
to propose a competitive price. In other  words, he has to take the risk of
losing money, with a small probability, in order to have the 
opportunity to make money with the contract (selling  the option and hedging it). 
We quantify this risk thanks to a Value-at-Risk technique. In a similar way as in UVM, we 
find distinct buying and selling prices due to the asymmetry of risks between buyer
and seller.

In practice, the risk assessment is often carried out in terms of sensitivity with 
respect to a deterministic variation of model parameters. 
These sensitivities are usually called greeks and option sellers
try to manage their portfolios in order to control the risk related to these sensitivities 
using risk measures like historical Value-at-Risk.
However, this approach leads to mathematical difficulties, for instance in dealing with
infinite dimensions, and practical difficulties. 

An alternative method that we will use here is based on potential theory 
(see Albeverio \cite{bib:Albeverio}, Bouleau and Hirsch
\cite{bib:Bouleau-Hirsch} and Fukushima et al.
\cite{bib:Fukushima}) suggested by Bouleau
\cite{bib:Bouleau-erreur}. 
We assume that the uncertainties are small random variables added 
to the true values of the parameters. The laws of these estimators are generally 
poorly known,  so we assume that only the two first moments are known, i.e. the biases 
and the variances/covariances. 
When the uncertainty on  a parameter is small, it is reasonable to neglect  high moments.
This is also justified by the fact that parameter estimation is often difficult owing to 
the non-linearity, the complexity of the equations and to the poor quality or the shortage 
of data.
The theory of Dirichlet forms leads to an "error theory" allowing the computation of the 
two first moments of the random variable that describes the option seller's profit or lost  
due to the parameter uncertainty. 
This theory combines the variance to the ``carr\'{e} du champ'' operator and the
bias to the generator  associated to the Dirichlet form, see Bouleau  
\cite{bib:Bouleau-erreur2}.

To summarize, the main result of this paper is to give closed forms for the corrections on 
the price of contingent claims due to the uncertainties on the parameters that are modeled 
as statistical estimators, i.e. random variables, with small variances and biases. 
These are corrections for an investor seeking to compensate the bias on the option 
price and  accepting a residual risk, both due to the  uncertainty on parameters. We define 
the  risk taken by the option seller in a statistical way in accord with a VaR risk measure.
A direct consequence of this pricing method is the separation of the buy and sell prices due 
to the asymmetry between the buyer and the seller with respect to the risk on parameters.
A surprising result is that a systematic bias exists with respect to the pricing without uncertainty 
even if the parameters are unbiased. 
Indeed, this bias is due to the non-linearity of the payoff.
Our analysis covers a large class of stochastic process with continuous paths and,
in particular, equity models with local or stochastic volatility. 

Finally, we present an example based on the Black-Scholes model, with uncertainty on 
the volatility parameter. We show that that our analysis leads to a smiled implied volatility surface
and we give the explicit formula for vanilla option prices and the bid-ask spread. 

The paper is organized as follows.
In Section 2, we   resume the notation used in the paper. In Section 3, we
briefly recall the error theory based on Dirichlet forms techniques. 
Section 4 deals with the problem of uncertainty estimation while
Section 5 is devoted to the study of the impact of uncertainty on the diffusion model. 
We analyze the profit and loss process and compute its law depending on both
the underlying diffusion and the parameter uncertainty. 
In Section 6,  we introduce a 
pricing principle to over-hedge the contingent claim and exhibit the bid and ask prices.
In Section 7, we give an example based on log-normal diffusion without drift. We exhibit
the bid and the ask prices and we prove, under certain hypotheses, that the implied volatility
exhibits a smile behavior.
Section 8 presents a second detailed example based on constant elasticity of variance model
in order to show the powerful of the method in a practical case.


\section{Notation}

This section fix the notation used throughout the paper. We consider two probability spaces 
$(\Omega_1, \mathcal{F}^1, \{\mathcal{F}^1_t\}_{t\in[0,T]}, \mathbb{P}_1 )$ and 
$(\widetilde{\Omega}, \widetilde{\mathcal{F}}, \widetilde{\mathbb{P}})$.
The reference filtration $\{\mathcal{F}^1_t\}_{t\in[0,T]}$ of the first probability space 
$(\Omega_1, \mathcal{F}^1, \{\mathcal{F}^1_t\}_{t\in[0,T]}, \mathbb{P}_1 )$ satisfies the usual condition and
is generated by a Brownian motion $B$. In this space, we define a continuous time financial market 
model consisting of a risk-free asset, whose price process is assumed for simplicity to be equal to 
$1$ at any time, and a risky asset (or stock) of price process $S = (S_t)_{t \geq 0}$.  We will denote by $\mathbb{E}_1$
the expectation under $\mathbb{P}_1$. 

The second probability space $(\widetilde{\Omega}, \widetilde{\mathcal{F}}, \widetilde{\mathbb{P}})$ 
is the sample space used to define the parameter estimators. We will denote by small letters, e.g $a$, the true parameters 
of the model; by capital letters, e.g. $A$, the estimators of these parameters and by  capped capital letters, e.g. $\widehat{A}$,
the estimated values of these parameters. But, to avoid any confusion, the true volatility value will be denoted by $\sigma$,
its estimator by $\varsigma$ and its estimated value by $\widehat{\varsigma}$.
As a consequence, $A$ will be a random variable defined in the space  $(\widetilde{\Omega}, \widetilde{\mathcal{F}}, \widetilde{\mathbb{P}})$, when $\widehat{A}$ will be a fixed value resulting from a particular observed dataset.

We consider the product space 
$(\Omega, \mathcal{F}, \mathbb{P}) = (\Omega_1 \times \widetilde{\Omega}, \mathcal{F}^1 
\otimes \widetilde{\mathcal{F}}, \mathbb{P}_1 \times \widetilde{\mathbb{P}})$.
We assume that the Brownian motion $B$ remains a Brownian motion under the probability space 
 $(\Omega, \mathcal{F}, \mathbb{P})$ and that all estimators defined in the probability space 
 $(\widetilde{\Omega}, \widetilde{\mathcal{F}}, \widetilde{\mathbb{P}})$ are independent with respect to the 
 filtration generated by the Brownian motion.

 We need to introduce an auxiliary probability space denoted by $(\overline{\Omega}, 
 \overline{\mathcal{F}}, \overline{\mathbb{P}})$. This space is a copy of estimator probability space 
 $(\widetilde{\Omega}, \widetilde{\mathcal{F}}, \widetilde{\mathbb{P}})$. We will denote by $\overline{\mathbb{E}}$
 the expectation under the probability $\overline{\mathbb{P}}$.
This auxiliary probability space is introduced only to perform mathematical proofs and has no influence on final results.

Finally, we introduce a representation basis, denoted by $\{\phi_i(t,x)\}_{i\in \mathcal{N}}$, of the space of volatility function.
All functions $\phi_i$ belong to $C^{1,2}$, $\{\phi_i\}_{i\in \mathbb{N}}$ and their derivatives are  Lipschitz by a same constant 
and square integrable in x.
The space of admissible volatility functions, denoted by $L^2_{\sigma}$, is     
$$
L^2_{\sigma} = \left\{f(t, x) = \sum_i b_i \phi_i(t,x) \text{ such that }  (b_1, \ldots, b_n, \ldots) \in \ell^2 \text{ and } f(t,x) > \xi >0\; 
\forall (t,x)  \right\}
$$ 
We equip the space with the distance $d^2_\sigma$ induced by the sequence representation.

\section{Mathematical tools}\label{sec:error-theory}

We begin by giving a general introduction to the study of sensitivity with respect to a 
stochastic perturbation and a formal definition of the framework that we will use, i.e. 
error theory based on Dirichlet forms, in accord with  Bouleau \cite{bib:Bouleau-erreur}. 
In this survey, we follow \cite{bib:Bouleau-erreur2}.

We consider a function $F(u_1,\, u_2, \, ...,\, u_n)$ depending on the parameters
$u= (u_1,\, u_2, \, ...,\, u_n)$ that we assume approximately known with some 
uncertainty. 
That is, we assume that we have a statistical estimator $U= (U_1,\, U_2, \, ..., \, U_n)$ 
of the vector  $u$. We consider that the law of $U$ is unknown
 and we suppose known only the  variance and bias.
 We assume that the function $F$ is sufficiently regular and we seek to evaluate 
 the impact on $F$ of the uncertainties on $u$.  
 The first study of this problem goes back to Gauss who proved the following expansion 
 for the variance of $F(U)$ if the uncertainties are small compared with the parameters values:
 $$
 \text{Var}[F(U)] = \sum_{i,\, j}^n \frac{\partial F}{\partial u_i}(u) \,  
 \frac{\partial F}{\partial u_j}(u) \, \text{Covar}[U_i,\, U_j] 
 $$
However, this relation is valid only if the number of  parameters is fixed and 
if $F$ has an explicit formula, in particular,  it cannot be defined via a limit.
Another main problem is that we do not know $u$ but only $U$, i.e. the estimator of $u$. As a consequence, we cannot evaluate the derivatives of $F$ at $u$.  Some natural questions that arise 
are: can we use the same approach if we know only the random variable $U$? That is, can we
prove that
 $$
 \text{Var}[F(U)] = \sum_{i,\, j}^n \frac{\partial F}{\partial u_i}(U) \,  
 \frac{\partial F}{\partial u_j}(U) \, \text{Covar}[U_i,\, U_j] \; \,?
 $$ 
Is it this relation true when $F$ is defined implicitly, for instance by a $L^2$-limit?
Can we study the bias of $F(U)$ caused by a non-linear function $F$?
The classical Gauss theory does not answer to all of these problems.
Therefore, we turn us to  the error theory based on Dirichlet forms in order to go beyond these problems. In particular, we need to evaluate 
the effect of a stochastic error on more complex objects such as stochastic integrals. 
For this, we consider an estimated value $\widehat{U}$ of a parameter $u$ with a small uncertainty $ Y$, 
on which we compute a non-linear function $F$. Our estimate $\widehat{U}$  is replaced by 
the estimator $U=\widehat{U} + Y$, i.e. a random variable 
with a mean squared error $\text{MSE}[Y]$ and bias  $\text{Bias}[Y]$. 
Appling Taylor's expansion to a function $F$, we find
\begin{eqnarray*}
\text{MSE}[F(U)] & \equiv & \widetilde{\mathbb{E}}\left[ \{F(U)-F(\widehat{U})\}^2 \right] = 
\left[F^{\prime}(\widehat{U})\right]^2 \; \text{MSE}[U] 
+  \text{higher orders} \\
\text{Bias}[F(U)] & \equiv & \widetilde{\mathbb{E}}[F(U)-F(\widehat{U})] = F^{\prime}(\widehat{U}) \; \text{Bias}[U] + \frac{1}{2}
 F^{\prime \prime}(\widehat{U}) \; \text{MSE}[U] + \text{higher orders} \; .
\end{eqnarray*} 
Clearly, the mean squared error $\text{MSE}[F(U)]$ dominates the variance $\text{Var}[F(U)]$ defined as 
$\widetilde{\mathbb{E}}\left[ \{F(U)- \widetilde{\mathbb{E}}[F(\widehat{U})] \}^2 \right]$. For the centered random variable 
$U$ the mean squared error and the variance take the same value.
Moreover, the difference $\text{MSE}[F(U)]- \text{Var}[F(U)]$ is equal to the square of the bias, that is, it is negligible in our expansion. 
As a consequence, we will use the mean squared error as a conservative estimation of the variance. 

If we suppose the uncertainty small, we can  truncate the Taylor's expansion and find two closed 
transport formulas for the bias and the mean squared error. These are known in literature, since the 
mean squared error has the same transport formula as a "carr\'{e} du champ" operator on a probability space
equipped with a local Dirichlet form, while the bias has the same transport formula as 
the generator  of the semigroup associated to the Dirichlet form, 
see for instance Bouleau and Hirsch \cite{bib:Bouleau-Hirsch}. 
The main advantage of this comparison is that the carr\'{e} du champ operator  and the 
generator of  the semigroup are closed  operators with respect to the graph norm: see 
for instance Fukushima et al. 
\cite{bib:Fukushima}. Therefore, they are good operators for studying objects defined 
by limits, such as stochastic integrals. 

The axiomatization of this idea was introduced by Bouleau \cite{bib:Bouleau-erreur} as follows: He defined an error structure as a probability space equipped with a local Dirichlet form with a  
carr\'{e} du champ operator.

\begin{definizione}[Error structure]\hfill
\vspace{0.2cm}

An error structure is a term
$\displaystyle  \left( \widetilde{\Omega}, \, \widetilde{\mathcal{F}}, \, \widetilde{\mathbb{P}}, \, \mathbb{D}, \, \Gamma \right)$, where

\begin{itemize}
\item{$\left( \widetilde{\Omega}, \, \widetilde{\mathcal{F}}, \,
\widetilde{\mathbb{P}} \right)$ is a probability space;}
 \item{$\mathbb{D}$ is a dense sub-vector space of $L^2\left(
\widetilde{\Omega}, \, \widetilde{\mathcal{F}},
      \, \widetilde{\mathbb{P}}\right)$;}
\item{$\Gamma$ is a positive symmetric bilinear function from
$\mathbb{D} \, \times
    \, \mathbb{D}$ into $L^1 \left( \widetilde{\Omega}, \, \widetilde{\mathcal{F}},
      \, \widetilde{\mathbb{P}} \right)$ satisfying the functional calculus of class
    $\mathcal{C}^1 \cap Lip$, i.e. if F and G are of class $\mathcal{C}^1$ and globally Lipschitz, 
    and U and V belong to $\mathbb{D}$, then F(U) and G(V) belong  to $\mathbb{D}$ and}
\begin{equation}\label{functional-calculus}
  \Gamma\left[F(U), \, G(V) \right] = F'(U) \,G'(V) \,  \Gamma[U, \, V] \; \;
  \widetilde{\mathbb{P}} \; a.s.; 
  \end{equation}
\item{the bilinear form $\mathcal{E}[U, \, V] = \frac{1}{2}
\widetilde{\mathbb{E}}\left[\Gamma[U,
      \, V]\right]$ is closed;}

\item the constant $1$ belongs to $\mathbb{D}$ and $\mathcal{E}[1,1]=0$
\end{itemize}

\end{definizione}

We generally write $\Gamma[U]$ for $\Gamma[U,\, U]$ and we  denote by 
$\Gamma[U](\widehat{U})$ the particular realization of the random variable $\Gamma[U]$ corresponding to observed data, 
i.e. $\Gamma[U](\widehat{U})$ is the conditional variance of the estimator $U$ given the particular dataset used to estimate it 
and such that the estimated value of $u$ is equal to $\widehat{U}$.
With this definition, $\mathcal{E}$ is a Dirichlet form and $\Gamma$ is the associated carr\'{e} du champ operator. 
The Hille-Yosida theorem 
guarantees that there exists a semigroup and a generator $\mathcal{A}$ that are coherent 
with the Dirichlet form  
$\mathcal{E}$, see for instance Albeverio \cite{bib:Albeverio} and Fukushima et al. \cite{bib:Fukushima}. 
This generator $ \mathcal{A}\,: \, \mathcal{D}  \mathcal{A} 
\rightarrow L^1(\widetilde{\mathbb{P}})$  is a 
self-adjoint operator, its domain $\mathcal{D}  \mathcal{A}$ is included into $\mathbb{D}$, and  
this operator satisfies, for $F \in \mathcal{C}^2$, $U \in \mathcal{D}  \mathcal{A} $ 
and $\Gamma[U] \in L^2(\widetilde{\mathbb{P}})$:
\begin{equation}\label{bias-chain-rule}
  \mathcal{A}\left[F(U) \right] = F'(U)\,  \mathcal{A}[U] + \frac{1}{2} F''(U)\,  \Gamma[U] \; \;
  \widetilde{\mathbb{P}} \; a.s.\,.
\end{equation}
Moreover, it is a closed operator with respect to the graph norm.  We denote by $\mathcal{A}[U](\widehat{U})$ 
the particular realization of the random variable $\mathcal{A}[U]$ corresponding to observed data,
i.e. $\mathcal{A}[U](\widehat{U})$ is the bias of the estimator $U$ given the particular dataset used to estimate it 
and such that the estimated value of $u$ is equal to $\widehat{U}$. 
We emphasize two important points related to this theory. 
First of all, error structures have nice properties: in particular, it is possible to prove 
that the product of two or countably many error structures is an error structure (see 
Bouleau \cite{bib:Bouleau-erreur2}). 
Furthermore, the elements of the space $\mathbb{D}$ are square integrable random variables.
Therefore, this theory is coherent with the stochastic nature of the estimators $U$, and we consider the operator $\mathcal{A}$ and $\Gamma$ as the conditional bias and mean squared error given the value of the estimator $U$: for details we refer to  Bouleau
\cite{bib:Bouleau-erreur}, chapters III and V, and
\cite{bib:Bouleau-erreur2}.

The classical example of an error structure is the Ornstein-Uhlenbeck one: 

\begin{definizione}[Ornstein-Uhlenbeck error structure]\label{structureOU}\hfill
\vspace{0.2cm}

The Ornstein-Uhlenbeck structure is $\displaystyle  \left( \mathbb{R}, \, \mathcal{B}(\mathbb{R}), \, \widetilde{\mathbb{G}}, \, 
\mathbf{H}^1, \, U\rightarrow (U^{\prime})^2 \right)$, where $ \widetilde{\mathbb{G}}$ represent a unidimensional Gaussian law, 
$\mathbf{H}^1$ is the Sobolev space of order 1 and the carr\'{e} du champ operator is just the square of the first derivative.

\end{definizione}

The main drawback of the carr\'{e} du champ operator is its bi-linearity, which makes 
computations awkward to perform. An easy way to overcome this drawback
is to introduce a new operator, called the sharp operator. This operator is a specified choice for the gradient 
in the terminology of Bouleau and Hirsch (see 
Bouleau and Hirsch \cite{bib:Bouleau-Hirsch}, section II.6.).
We recall the definition of the sharp operator associated with $\Gamma$. 

\begin{proposizione}[Sharp operator]\label{prop:sharp}\hfill
\vspace{0.2cm}

Let $\left( \widetilde{\Omega}, \, \widetilde{\mathcal{F}}, \,
\widetilde{\mathbb{P}}, \, \mathbb{D}, \, \Gamma \right)$ be an error
structure and let $\left( \overline{\Omega}, \, \overline{\mathcal{F}},
\, \overline{\mathbb{P}} \right) $ be a copy of the probability space
$\left( \widetilde{\Omega}, \, \widetilde{\mathcal{F}}, \,
\widetilde{\mathbb{P}}\right) $. We assume 
that the space $\mathbb{D}$ is separable. Then there exists a linear operator, called
sharp and denoted by $(\,)^{\#}: \, \mathbb{D} \rightarrow L^2(\widetilde{\mathbb{P}} \times \overline{\mathbb{P}})$, with the following two properties:

\begin{itemize}

\item{$\forall \, U \in \mathbb{D}$, $\Gamma[U] =
\overline{\mathbb{E}}\left[\left(U^{\#}\right)^2\right]$, where $\overline{\mathbb{E}}$ denotes the expectation under the probability $\overline{\mathbb{P}}$;}

\item{$\forall \, U \in \mathbb{D}^n$ and $F \in \mathcal{C}^1 \cap
Lip$,  $\left(F(U_1, \, ... \,, \, U_n)\right)^{\#} =
 \sum_{i=1}^n \left(\frac{\partial F}{\partial x_i}\circ U \right) U_i^{\#}  $.}

\end{itemize}
\end{proposizione}

The sharp operator is a useful tool when computing $\Gamma$ because 
it is linear, whereas the carr\'{e} du champ operator is bilinear.  Moreover, the sharp operator is closed, 
then, for instance, we can exchange the sharp operator and the integral sign. The proof of this fact proceeds by an 
approximation of the integral by a sum, then we apply the sharp operator  and finally we take the limit using the 
closeness of the sharp operator, see for instance Bouleau  \cite{bib:Bouleau-erreur} section VI.2.

We give a detailed example of the computations of sharp, carr\'{e} du champ and bias operator in Section 8.   
To make our approach more understandable for a non specialist of Dirichlet forms, we explain in this example 
the computation of each operator. The classical approach is to compute the sharp operator. This operator is just an auxiliary 
step needed to the computation of the carr\'{e} du champ, that represents the variance in error theory framework.
Finally, the knowledge of variance operator $\Gamma$ allows us to compute the action of the bias operator $\mathcal{A}$.

We observe that the error theory based on Dirichlet forms
restricts its analysis to the study of the first two orders of error propagation, i.e. the  bias and the
variance. This fact is justified by the lack of information on the parameter uncertainties, 
generally given by the Fischer information matrix, that is often quite limited. 
The study of  higher orders is a very difficult problem for both mathematical and practical 
reasons. From the mathematical point of view, it would be necessary  to study chain rules 
of higher orders, involving skewness and kurtosis, and to prove that the related operators 
are closed in a suitable space. However, the crucial problem remains to have sufficiently
accurate estimates for the higher order uncertainties. This statistical obstacle cannot be 
overcome easily. Therefore, we decide to restrict our study to the two first orders. 

We conclude this survey recalling two direct consequences of this approach. 

\begin{corollario}[Impact of uncertainty]\label{remark-expansion}\hfill
\vspace{0.2cm}

If we neglect moments higher than the second one, i.e. we approximate all random variables by the Gaussian 
one with same expectation and variance, the impact of uncertainty on the parameter $u$ transforms $F(U)$ 
into a Gaussian random variable of the form
\begin{equation}
F(U) \stackrel{\text{\tiny d}}{\approx} F(\widehat{U}) + 
\mathcal{A}[F(U)](\widehat{U}) + 
\sqrt{\Gamma[F(U)](\widehat{U})} \; G\, ,
\end{equation}
where G is a standard Gaussian variable.

\end{corollario}

If we want to be conservative, in particular if we suppose that the Gaussian approximation is
not valid, then we can apply left-hand tailed Chebyshev's inequality (also known as Cantelli inequality), 
which leads to the following statement.

\begin{corollario}[left-hand tailed Chebyshev's inequality]\label{prop:Cheb}\hfill
\vspace{0.2cm}

The random variable $F(U)$ verifies the following inequality for all $k\geq 1$.
\begin{equation}\label{cheb-equation}
\widetilde{\mathbb{P}}\left[F(U)-F(\widehat{U}) - \mathcal{A}[F(U)](\widehat{U}) \geq k\,  \sqrt{
\Gamma[F(U)](\widehat{U})}  \right] \leq \frac{1}{1+k^2}
\end{equation}
\end{corollario} 

\section{Estimation, calibration and uncertainty}

In this section, we discuss briefly the classical methods to determine the model parameters and how to evaluate their uncertainties.

\smallskip
\underline{\emph{Statistical estimation.}} The classical approach deals with  estimating 
the parameters using historical data of the stock  price.
The goal of this method is to construct an estimator, that takes the historical data as 
input and gives as result an estimated value for each parameter.
Then, an estimator is a function mapping the historical data probability space 
(sample space throughout the sequel) into the set of possible estimated outcomes. 
It is important to remark that an estimator is a random variable whereas an estimated value is a constant. 
We can associate different attributes with an estimator, like its bias, its variance and its mean squared error. We recall 
that all these attributes depend both on the estimator, i.e. the rule to find the estimated value, and on the sample, 
i.e. the historical data. In this sense,  
maximum likelihood estimation plays a central role between statistical estimation methods. 
We recall the asymptotically efficiency  between their properties, 
meaning that the maximum likelihood estimator  reaches the Cramer-Rao bound,
that is a lower bound on the variance of a parameter estimator.
Bouleau and Chorro \cite{bib:Bouleau-Chorro} have proved that the operator 
$\Gamma$ is the inverse of the  Fisher information matrix, i.e. the inverse of Cramer-Rao bound, in the finite dimensional case. 

\smallskip
 \underline{\emph{Calibration on liquid options.}} The main drawback of statistical estimation 
 is that the diffusion process depending on the estimated parameters is 
 generally not consistent with the market prices of vanilla options. But the liquidity of these 
 options is large enough to consider that their prices represent equilibrium prices.
 As a consequence, a financial model has to be consistent with these prices. The method
  of calibration is to find which sets of parameters are consistent with market prices of 
 vanilla options. A general result due to Dupire \cite{bib:Dupire} shows for instance that a 
 unique local volatility model is coherent with the surface of option market prices across 
 stikes and maturities. However, Dupire formula is an ill posed problem. 
 A regularization method is then required in order to stabilize the solution. 
 The classical regularization method is to add a penalization term and to find a tradeoff 
 between consistency with data and smoothness, see for instance 
 Cont and Tankov \cite{bib:Cont}. 
 
 Calibration method gives a rule to find the best set of
  parameters, i.e. the set of parameters that minimizes a given optimization principle. 
 This rule is a function from the sample space into the possible outcomes, like the 
 estimator in the case of statistical estimation. 
 But it is difficult to adapt the previous methodology to calibration case, since the 
 calibration function from the sample space into the possible outcomes is not explicit.

A different approach can be however proposed. It is well-known that, even if option volumes 
are grown exponentially in the last thirteen years, their relative spread remains large compared with assets relative spread.
The relative spread is defined as the ratio between the bid-ask spread and the mid price.
As a consequence the equilibrium price of a vanilla option is known with an uncertainty, measured by the bid-ask spread.
This uncertainty has to play a role in the calibration methodology. 

We propose then an easy method to transfer these price uncertainties into an uncertainty on the calibrated parameters.
The first step is to perform the calibration methodology neglecting the bid-ask spread, 
i.e. for any option we fix the price at the mid-price.       
Afterward, we fix an option $j$ on the basket used to calibrate, we shift its price to ask price and we recalibrate. The new set 
of calibrated parameters represents a stress of the previous set. We compute then the difference between the two sets and we interpret it as a standard deviation of the calibrated parameters submitted to a stress into the price of option $j$. So we can reconstruct the variance-covariance matrix with respect to the random source in option $j$. 
We perform this procedure for all options into the basket and we can reconstruct the global variance-covariance matrix.

This method can be easily performed given that only a rough approximation of the stressed calibration set is needed. 
Moreover, the initial calibration set can be used as a starting set to perform stressed calibrations. 
A corollary result of this approach is to evaluate the fit  of a model to the real data, since a large variance-covariance 
is the sign of an overfitting situation due to a too large number of parameters or to an unconformability 
of the model behavior to real data shape.

\smallskip
\underline{\emph{Pure calibration method.}} Finally, a third method to define parameters uncertainty in our framework is 
via a pure calibration approach. In other word, we can study theoretically the behavior of the model depending both on the estimated parameters and their uncertainty. Given the pricing formula or the pricing rule, for instance the simulation using Monte Carlo, it is possible to study the parameters sets coherent with vanilla options.    

\bigskip

The crucial remark of this section is that statistical estimation and calibration produce a set of estimated parameters, that is the
best approximation of true parameters given all historical data. As a matter of fact, these best estimated values are both 
\begin{description}
\item[Fixed Values] since historical data are fixed and known. Then the rule to estimate, via statistical estimation or calibration, 
gives us a numerical result. 
\item[Random Variables] since historical data are in fact a given realization. Then  the dependence with respect to the simple
space cannot be neglected.   
\end{description}        
As a consequence, in the pricing methodology the estimated parameters has to be considered fixed and this set of parameters 
is the only one that can be used to perform all computation.
However, an estimation risk appears and the dependency on the given realization of historical data becomes the kernel 
question to be analyzed. 
The estimated volatility is then a fixed value (denoted by $\widehat{\varsigma}$) on the probability space 
$\left(\Omega_1, \, \mathcal{F}^1, \, \mathbb{P}_1\right)$ and a random variable 
(denoted by $\varsigma$) on the probability space $(\widetilde{\Omega}, 
 \widetilde{\mathcal{F}}, \widetilde{\mathbb{P}})$.

The main objective of the present paper is to investigate the impact of this double nature of estimated parameters set.

\section{Diffusion model}

We start our analysis with the classical Black Scholes model, hereafter denoted BS, 
see Black and Scholes \cite{bib:Black-Scholes}.  For the sake of simplicity, we take the
money market account as the numeraire.
Let $\left(\Omega_1, \, \mathcal{F}^1, \, \mathbb{P}_1\right)$ be the historical
probability space and $B$ the associated Brownian motion. Fix $\mu \in \mathbb{R}$,
$\sigma _0>0$ and the interval $[0,\,T]$. The dynamics of the risky asset under the historical probability measure 
$\mathbb{P}_1$ is given by the following diffusion in accord with the model of Black 
and Scholes:
\begin{equation}\label{Black-Scholes}
dS_t  =  S_t \, \mu \, dt + S_t \, \sigma_0 \, dB_t .
\end{equation}
In this framework, the price of a European vanilla option is now standard
(see for instance \cite{bib:Lamb-Lap}).
This model presents many advantages, in particular the pricing depends
only on the volatility parameter $\sigma_0$ and we find closed forms for premiums
and greeks of vanilla options.  However, the BS model cannot
reproduce the market price of call options for all strikes with the
same volatility: this is called the smile effect.
In order to take  this phenomenon into account, we discuss two main extensions, 
the local and the stochastic volatility models, hereafter denoted LV, see for instance 
Dupire \cite{bib:Dupire}, and SV, see Hull and White \cite{bib:Hull},
Heston \cite{bib:Heston}, Hagan et al \cite{bib:Hagan} 
and Fouque et al \cite{bib:Fouque}. 
In these two classes of models, the parameter $\sigma_0$ is replaced by a function 
$\sigma$ that depends on the time $t$, on the underlying $S$ and, in the case of 
SV models, on a random source. For the sake of simplicity, we consider a local volatility 
model. 
The stochastic differential equation verified by the 
price of the underlying is
\begin{equation}\label{SDE-SV}
dS_t = S_t \, \mu \, dt + S_t \, \sigma(t, \, S_t) \, dB_t \, .
\end{equation}
It is plain that this class of models is not the more general. But, your methodology can be easily 
applied to multidimensional diffusion, then the extension to SV models is possible under some 
hypotheses. 
Diffusions with jumps  are excluded from the models treatable with your approach instead; for these, see 
for instance Cont and Tankov \cite{bib:Cont}.
Henceforth, we denote by $(\Omega_1, \, \mathcal{F}^1,\, \mathbb{P}_1)$ the probability space 
where the Brownian motion $B$ is defined, by $\{\mathcal{F}^1_t\}_{t \in [0,T]}$ the standard filtration 
generated by the Brownian motion $B$ and  by $\mathbb{E}_1$ the expectation 
under the probability $\mathbb{P}_1$.   

Our goal is to analyze the sensitivities of the model given by the SDE  (\ref{SDE-SV})
with respect to the uncertainty in the estimation of their parameters. 
Indeed, all models depend on certain parameters, generally a small number, that are
related to the underlying, and this makes it possible to calibrate the model. 
When we select a given value for a parameter, using a calibration methodology, there
is still some uncertainty on the true value of this parameter. 
This uncertainty can be estimated by using statistical methods, such as Fischer information,
or by computing the sensitivity of the model calibration. We emphasize that both methods,
statistical and calibrative, yield parameters with uncertainties and these uncertainties 
have a random character.
 
Our goal is to analyze the impact of these uncertainties on the management and the 
hedging of a contingent claim from the point of view of a seller that tries to minimize
his own risk.  

We propose to consider the model given by the SDE (\ref{SDE-SV}) with uncertainty on the
parameter modeled  by means of an error 
structure on the volatility function $\sigma(t, \,S_t)$.

We make the following financial hypothesis:

\begin{assumption}[Asset evolution and uncertainty impact]\label{assumption-fin}\hfill

\begin{enumerate}
\item{ the market follows the SDE (\ref{SDE-SV})  with fixed but unknown parameters, 
i.e. the underlying follows the SDE (\ref{SDE-SV}). 
The market is viable and complete, i.e. there are enough 
traded assets to guarantee that any contingent claim admits  a hedging portfolio and there 
is a probability measure $\mathbb{Q}$, equivalent to $\mathbb{P}$, under which all 
discounted security prices are martingales;}

\item{the option seller knows that the underlying follows an SDE such as (\ref{SDE-SV}) 
but does not know  the values of the parameters of the function $\sigma$;}

\item{ the option seller has to estimate the parameters of his model: 
the uncertainty associated to such estimation is modeled  by means of an error structure. 
As a consequence, the price and the greeks of the option are affected by uncertainty. 
We assume that the option seller makes his statistical estimations before to sell 
the option and he does not change his estimators during the time interval $[0,\,T]$.}
\end{enumerate}
\end{assumption}

We assume that there exists a basis, denoted by $\{\phi_i(t,\, x)\}_{i \in \mathbb{N}}$, 
of the space of volatility functions.
The functions $\phi_i$ belong to $C^{1, \, 2}$ in $(t,\, x)$, , 
$\{\phi_i\}_{i\in \mathbb{N}}$ and their derivatives are  Lipschitz by a same constant 
and square integrable in x.
For instance, the basis  $\{\phi_i(t,\, x)\}_{i \in \mathbb{N}}$ can be the spanning set of 
interpolating splines for the volatility function.  
We define the set $L^2_{\sigma}$ of admissible volatility functions:
$$
L^2_{\sigma} = \left\{f(t, x) = \sum_i b_i \phi_i(t,x) \text{ such that }  (b_1,  b_2,  \ldots,  b_n, \ldots) \in \ell^2 \text{ and } f(t,x) > \xi >0\; 
\forall (t,x)  \right\}
$$ 
We equip the set with the distance $d^2_\sigma$ induced by the sequence representation.

We denote by $a_i$ the coefficients of the series expansion of $\sigma$, i.e.
$\sigma(t, \, x)=\sum_i a_i\,\phi_i(t, \, x)$. 
We assume that the option seller has to estimate the coefficients $a_i$, then the 
uncertainty is carried by the coefficients $a_i$ that are estimated 
by the estimators $A_i$, that is, the $A_i$
are random variables. The estimated values of the coefficients $a_i$ are denoted by $\widehat{A}_i$, i.e. $\widehat{A}_i$ 
are fixed values. We denote by $\varsigma(t,\, x)$ (resp. $\widehat{\varsigma}(t,\, x)$ ) the volatility estimator (resp. the estimated volatility), i.e.  $\varsigma(t, \, x)=\sum_i A_i\,\phi_i(t, \, x)$ (resp.
$\widehat{\varsigma}(t, \, x)=\sum_i \widehat{A}_i\,\phi_i(t, \, x)$ ).   
We make now the following mathematical hypothesis:

\begin{assumption}[Uncertainty on volatility]\label{assumption-math}\hfill
\vspace{0.2cm}

We assume $\sigma(t,x )$ and $\widehat{\varsigma}(t,x)$ belong to $L^2_{\sigma}$.
For each random variable $A_i$, we define moreover an independent error structure 
$\left( \mathbb{R}, \, \mathcal{B}(\mathbb{R}), \,
\widetilde{\mathbb{P}}_i, \, \mathbb{D}_i, \, \Gamma_i \right)$. We denote by 
$\left( \widetilde{\Omega}, \, \widetilde{\mathcal{F}}, \,
\widetilde{\mathbb{P}}, \, \mathbb{D}, \, \Gamma \right)$ the product of all error structures 
$\left( \mathbb{R}, \, \mathcal{B}(\mathbb{R}), \,
\widetilde{\mathbb{P}}_i, \, \mathbb{D}_i, \, \Gamma_i \right)$, that is again an error structure, see 
\cite{bib:Bouleau-erreur} chapter IV. We denote by $(\mathcal{A},
\, \mathcal{D}\mathcal{A})$ the related generator and its domain. We assume that the following 
hypotheses hold for each $A_i$.

\begin{enumerate}

\item $A_i \in \mathcal{D} \mathcal{A}$, $\Gamma[A_i]$ and $\mathcal{A}[A_i]$ are known;

\item{the error structure $\left( \widetilde{\Omega}, \, \widetilde{\mathcal{F}}, \,
\widetilde{\mathbb{P}}, \, \mathbb{D}, \, \Gamma \right)$
 admits  a sharp operator denoted by $( \cdot)^{\#}$, in the sense that each error structure $\left( \mathbb{R}, \, \mathcal{B}(\mathbb{R}), \,
\widetilde{\mathbb{P}}_i, \, \mathbb{D}_i, \, \Gamma_i \right)$ admits a sharp operator denoted by $( \cdot)^{\# }_i$
 and the sharp operator on the product error structure is the sum of the sharp operators of each sub-error structure.}
 
\end{enumerate}

Finally, we assume the technical property 4.2. in 
Bouleau \cite{bib:Bouleau-erreur} chapter V page 84. 

\end{assumption}

Property 4.2 in Bouleau \cite{bib:Bouleau-erreur} fix the probability measures 
$\widetilde{\mathbb{P}}_i$ as a mixture of a measure 
$\mu$ absolutely continuous w.r.t. Lebesque one and a Dirac mass. 
Assuming the sum of the weights  of the measure $\mu$ into the mixture is finite, then  
we can prove, see \cite{bib:Bouleau-erreur},   that the probability product measure 
$\widetilde{\mathbb{P}} = \otimes \widetilde{\mathbb{P}}_i$ is absolutely 
continuous and under $\widetilde{\mathbb{P}}$  only a finite number of the coefficients 
$A_i$ in the representation are  random variables.  
That means that only a finite, but random, number of terms contribute to the uncertainty.

We briefly comment on the two previous assumptions. We have split them according to
the nature of the assumption: the first one, \ref{assumption-fin}, is mainly 
financial and the second one is mainly
mathematical. Assumption \ref{assumption-math} sets the mathematical framework used 
in this paper; point 4 and the final assumption are included in order to simplify proofs 
and to simplify some computations, but they can be relaxed. The hypotheses on the functions 
$\phi_i(t,\, x)$ are requisite to guarantee the existence and the uniqueness of the solution of SDE 
(\ref{SDE-SV}). The lower bound $\xi$ guarantees that the risk neutral probability $\mathbb{Q}_1$
is well defined, in particular it is equivalent to $\mathbb{P}_1$.

The main financial hypotheses are contained in assumption \ref{assumption-fin}. The first point 
states that the underlying follows an exact SDE without uncertainty. The uncertainty appears 
when the option seller attempts to determine the parameters of the SDE. 
This assumption is easy to understand in the financial framework and circumvents many 
problems in the mathematical framework. In particular, when the volatility is uncertain, 
the set of probability measures 
describing the whole class of possible probabilistic viewpoints is not  dominated. 
In this case, the classical approach used in mathematical finance (see for instance 
Avellaneda \cite{bib:Avelaneda}) cannot be followed; some important attempts to address 
this problem are Denis and Martini \cite{bib:Denis-1} and Denis and 
Kervarec \cite{bib:Denis-2}.   

When the option seller has set his parameters, i.e. his function $\widehat{\varsigma}(t,\, x)$, he has
established his "risk neutral" probability $\mathbb{Q}_1^{\widehat{\varsigma}}$, that is, the probability 
measure under which the diffusion 
\begin{equation}\label{SDE:perturbed-av}
dX^{(\widehat{\varsigma})}_t  = \mu \,X^{(\widehat{\varsigma})}_t\, dt +  
X^{(\widehat{\varsigma})}_t  \, \widehat{\varsigma} \left(t, \, X^{(\widehat{\varsigma})}_t \right)  \, dB_t 
\end{equation}
 is a  martingale, the related Brownian motion is denoted by $W_t^{(\widehat{\varsigma})}$.
 The probability $\mathbb{Q}_1^{\widehat{\varsigma}}$ exists and is unique since the risk premium is given by 
 $\frac{\mu}{\widehat{\varsigma}}$. We recall that the estimated volatility is a fixed function depending on time and on 
 the underlying $S$, then it is known at any time given the value of the spot. Moreover, the estimated volatility 
 $\widehat{\varsigma}$ is bounded from below by the positive constant $\xi$, then the Radon-Nikodym density is 
a martingale thanks to the Novikov criterion and the inequality
$$
\mathbb{E}_1 \left[  exp\left\{\frac{1}{2} \int_0^T \frac{\mu^2}{\widehat{\varsigma}^{2}}  dt \right\}\right]  < 
\mathbb{E}_1 \left[  exp\left\{\frac{1}{2}\int_0^T \frac{\mu^2}{ \xi^{2}} dt\right\} \right] <\infty \, .
$$
 Thanks to this probability, he can calculate the option price and the related hedging strategy.
 As a matter of fact, the replacement of the volatility function estimator $\varsigma(t,\, x)$ with the estimated value of 
 the volatility $\widehat{\varsigma}(t, \,x)$ causes an information loss.
 The natural question that arises is: What is the impact on the prices of the uncertainty on the
 volatility function $\varsigma(t,\, x)$ and it is possible to bound the related risk?
 
 We will analyze the problem in the following section \ref{sec:PandL} and we will propose a
 pricing principle in section \ref{sec:OP}.
We shall also make the following assumption on the class of contingent claims analyzed 
in this paper.

\begin{assumption}[Contingent claims]\label{assumption-payoff}\hfill
\vspace{0.2cm}

Let $\Phi$ be the payoff of a considered contingent claim. 
We assume that $\Phi$ depends only on $S_T$ and belongs to $C^2$. 
We also suppose that the two first derivatives of $\Phi$ are bounded.
\end{assumption}

This hypothesis is needed in order to apply error theory techniques. However, we will show at the end 
of section  \ref{sec:OP} that this hypothesis can be overcome. Another point where the hypothesis about 
the boundedness of the first derivative of $\Phi$ is crucial  is to prove that the gain process of the option seller
is a $\mathbb{Q}_1^{\widehat{\varsigma}}$-martingale.

\subsection{Profit and loss process}\label{sec:PandL}

We consider  that the option seller uses his data and proprietary information to determine
the model parameters using an optimization procedure, see for instance 
Dupire \cite{bib:Dupire}  and Cont and Tankov \cite{bib:Cont}. 
By Assumption \ref{assumption-fin}, a risk neutral-measure $\mathbb{Q}_1$ 
exists and is unique.  With this probability measure, the seller could define the fair price 
of the option and determine the hedging strategy if he  knew  the true values of the parameters. 
Since he does not know these true values, he has to estimate them before to sell the option in 
order to propose a price for the option and to define an hedging strategy.

The natural question is:
how can he define this price taking into account the uncertainty into his estimated parameters?
To do that, we recall that the wealth minus the liability of the option seller, that sells the option and hedges it, is the
crucial stochastic process to define the price of the option, since the fair price and the right 
hedging strategy are such as the wealth minus the liability of the option seller is  worth zero almost surely in a 
complete market without uncertainty.
Therefore we study this wealth reduced by the related liability also know as the profit and loss process.

We neglect, in this first work, the fact that the option holder can sell the option before  maturity, 
so we assume that he holds the contingent claim until  maturity. 
We also assume that all prices are denominated using the risk-free asset as numeraire.   
The profit and loss process at maturity of the option seller is given by the cost of the 
hedging strategy of the option, i.e. the option price in a financial model without uncertainty 
(see for instance Black and Scholes \cite{bib:Black-Scholes}), 
plus the value of the hedging portfolio minus the final payoff that the seller has to pay to the holder, 
so we have
\begin{equation}\label{equation:P-and-L}
P\&L(T) = C(\widehat{\varsigma}, \, x, \, 0) + \int_0^T \Delta(\widehat{\varsigma}, \, S_t, \, t) \, dS_t - \Phi(S_T),
\end{equation}
where $C$ and $\Delta$ denote respectively  the cost of the hedging strategy followed
by the option seller, and its first derivative with respect to the underlying. Thanks to the 
asset pricing theory and in absence of parameter uncertainty, we have 
$C(\sigma, \, x, \, 0) = \mathbb{E}_1^{\mathbb{Q}_1} \left[ \Phi(S_T) \right]$. However, since 
the option seller does not known the true values of the parameters, the cost of hedging 
strategy is given 
by 
$$
C(\widehat{\varsigma}, \, x, \, 0) =  \mathbb{E}^{\mathbb{Q}_1^{\widehat{\varsigma}}}_{X^{(\widehat{\varsigma})}_0=x}
\left[ \Phi \left(X^{(\widehat{\varsigma})}_T\right) \right]\, ,
$$
where  $ \mathbb{E}^{\mathbb{Q}_1^{\widehat{\varsigma}}}_{\,X^{(\widehat{\varsigma})}_0=x}$ denotes the expectation under the probability space
$(\Omega_1, \mathcal{F}^1, \{\mathcal{F}^1_t\}_{t\in[0,T]}, \mathbb{Q}_1^{\widehat{\varsigma}} )$ when the process $X^{(\widehat{\varsigma})}$ starts 
at time $0$ with value $x$ and follows the SDE (\ref{SDE:perturbed-av}), i.e.
\begin{equation}\label{SDE:perturbed}
dX^{(\widehat{\varsigma})}_t  = X^{(\widehat{\varsigma})}_t  \, \widehat{\varsigma} \left(t, \, X^{(\widehat{\varsigma})}_t
\right)  \, dW_t^{(\widehat{\varsigma})} 
\end{equation}
under probability $\mathbb{Q}_1^{\widehat{\varsigma}}$.
 To simplify our notations,
we denote by $\displaystyle \frac{\partial C}{ \partial \sigma(\phi_i)}$ and 
$\displaystyle \frac{\partial \Delta}{ \partial \sigma(\phi_i)}$  
the Gateaux derivatives of $C$ and $\Delta$ with
respect to the volatility in the direction of $\phi$, that is:
\begin{eqnarray}
\displaystyle \frac{\partial C}{ \partial \sigma(\phi_i)}(\widehat{\varsigma}, \, x, \, 0) 
&= & \displaystyle \lim_{\epsilon \rightarrow 0}
\frac{\mathbb{E}^{\mathbb{Q}_1^{\widehat{\varsigma}  +\epsilon \phi_i}}_{X^{(\widehat{\varsigma}+ \epsilon \phi_i)}_0=x} 
\left[ \Phi \left( X^{(\widehat{\varsigma}+ \epsilon \, \phi_i)}_T \right) \right] - 
\mathbb{E}^{\mathbb{Q}_1^{\widehat{\varsigma}}}_{X^{(\widehat{\varsigma})}_0=x} 
\left[ \Phi \left(X^{(\widehat{\varsigma})}_T \right) \right] }{ \epsilon} \label{rel:gateaux-1}
  \\
\displaystyle \frac{\partial \Delta}{ \partial \sigma(\phi_i)}(\widehat{\varsigma}, \, S_t, \, t) & = & 
\displaystyle \lim_{\epsilon \rightarrow 0} 
\frac{ \displaystyle   \frac{\partial \,   \mathbb{E}^{\mathbb{Q}_1^{\widehat{\varsigma}  +\epsilon \phi_i}
}_{X^{(\widehat{\varsigma}+ \epsilon \phi_i)}_t=x} 
\left[ \Phi \left(X^{(\widehat{\varsigma}+ \epsilon \phi_i)}_T\right) \right]}{\partial x} -
 \displaystyle   \frac{\partial \,   \mathbb{E}^{\mathbb{Q}_1^{\widehat{\varsigma}}}_{X^{(\widehat{\varsigma})}_t=x} 
\left[ \Phi \left(X^{(\widehat{\varsigma})}_T\right) \right]}{\partial x} 
}{\epsilon} \Bigg|_{x=S_t} . \label{rel:gateaux-2}
\end{eqnarray}
The next lemma shows that the two previous Gateaux derivatives exist.

\begin{lemmas}[Existence of the Gateaux derivatives]\label{Lemma-Gateaux}\hfill
\vspace{0.2cm}

Under assumptions \ref{assumption-math} and \ref{assumption-payoff}, the Gateaux derivatives defined by relations (\ref{rel:gateaux-1}) and (\ref{rel:gateaux-2})  exist. Moreover, the higher order Gateaux derivatives 
$\displaystyle 
\frac{\partial^2 C}{ \partial \sigma(\phi_i) \partial \sigma(\phi_j)}$ 
and  $\displaystyle \frac{\partial^2 \Delta}{ \partial \sigma(\phi_i) \partial \sigma(\phi_j)}$ exist too.
\end{lemmas}

\begin{proof}
We remark that $X_t^{(\widehat{\varsigma})}$ is a $\mathbb{Q}_1^{\widehat{\varsigma}}$ martingale. 
Then, the process $X_t^{(\widehat{\varsigma})}$ under $\mathbb{Q}_1^{\widehat{\varsigma}}$ has the same law as 
 $\overline{X}_t^{(\widehat{\varsigma})}$ under $\mathbb{Q}_1$, where  $\overline{X}_t^{(\widehat{\varsigma})}$ follows the SDE:
$$
d\overline{X}^{(\widehat{\varsigma})}_t  =
 \overline{X}^{(\widehat{\varsigma})}_t  \; \widehat{\varsigma} \left(t, \, \overline{X}^{(\widehat{\varsigma})}_t
 \right)  \, dW_t 
$$
where $W$ is a $\mathbb{Q}_1$ Brownian motion. Therefore, we can write the difference quotient
of $C(\widehat{\varsigma}, \, x,\, 0)$ with respect to a variation of the volatility $\widehat{\varsigma}$ 
in the direction $\phi_i$
in equation (\ref{rel:gateaux-1}) in the following way
$$
 \frac{\delta C}{\delta \sigma(\phi_i)}(\widehat{\varsigma}, \, x,\, 0, \, \epsilon) =  \frac{\mathbb{E}^{\mathbb{Q}_1}_{
 \overline{X}^{(\widehat{\varsigma}+ \epsilon \phi_i)}_0=x} 
\left[ \Phi \left( \overline{X}^{(\widehat{\varsigma}+ \epsilon \phi_i)}_T \right) \right] - 
\mathbb{E}^{\mathbb{Q}_1}_{\overline{X}^{(\widehat{\varsigma})}_0=x} 
\left[ \Phi \left(\overline{X}^{(\widehat{\varsigma})}_T \right) \right] }{ \epsilon} \,.
$$
Using  Assumption \ref{assumption-payoff}, we expand  the function $\Phi$ and it exists 
$$
J^{(\widehat{\varsigma}, \,\phi_i, \, \epsilon)}_T  \in  \left]  \min \left( \overline{X}^{(\widehat{\varsigma}+ \epsilon \, \phi_i)}_T,
\, \overline{X}^{(\widehat{\varsigma})}_T\right), \, \max \left( \overline{X}^{(\widehat{\varsigma}+ \epsilon \, \phi_i)}_T,\, 
\overline{X}^{(\widehat{\varsigma})}_T \right) \right[ \, ,
$$
such that
\begin{equation}\label{tech-eq-lemma-1}
 \frac{\delta C}{\delta \sigma(\phi_i)}(\widehat{\varsigma}, \, x,\, 0, \, \epsilon) =  
 \frac{\mathbb{E}^{\mathbb{Q}_1}
\left[ \Phi^{\prime} \left( \overline{X}^{(\widehat{\varsigma})}_T \right) \, Y_T^{(\widehat{\varsigma}, \, \phi_i, \epsilon)}  + 
\frac{1}{2} \Phi^{\prime\prime} \left(J^{(\widehat{\varsigma}, \,\phi_i, \, \epsilon)}_T \right) \, 
 \left(Y_T^{(\widehat{\varsigma}, \, \phi_i, \epsilon)} \right)^2   \right] }{ \epsilon} \, ,
\end{equation}
where the expectation is taken assuming the starting condition  
$\overline{X}^{(\widehat{\varsigma}+ \epsilon \, \phi_i)}_0  =\overline{X}^{(\widehat{\varsigma})}_0 =x$ and 
\begin{eqnarray*}
Y_T^{(\widehat{\varsigma}, \, \phi_i, \epsilon)} & = & \overline{X}^{(\widehat{\varsigma}+ \epsilon \, \phi_i)}_T-
\overline{X}^{(\widehat{\varsigma})}_T \, .
\end{eqnarray*}
We study the SDE satisfied by $Y_T^{(\widehat{\varsigma}, \, \phi_i, \epsilon)}$:
\begin{eqnarray*}
d Y_t^{(\widehat{\varsigma}, \, \phi_i, \epsilon)} &= & \overline{X}^{(\widehat{\varsigma}+ \epsilon \, \phi_i)}_t \left[
\widehat{\varsigma} \left(t, \, \overline{X}^{(\widehat{\varsigma}+ \epsilon \, \phi_i)}_t \right) + \epsilon
\phi_i \left(t, \, \overline{X}^{(\widehat{\varsigma}+ \epsilon \, \phi_i)}_t \right) \right]   \,dW_t \\
& & - \overline{X}^{(\widehat{\varsigma})}_t   \widehat{\varsigma} \left(t, \,\overline{X}^{(\widehat{\varsigma})}_t\right)
 \, dW_t \, ,
\end{eqnarray*}
with $Y_0^{(\widehat{\varsigma}, \, \phi_i, \epsilon)}=0$. The previous SDE can be written as
\begin{eqnarray*}
d Y_t^{(\widehat{\varsigma}, \, \phi_i, \epsilon)} &= & \epsilon \, 
\overline{X}^{(\widehat{\varsigma}+ \epsilon \, \phi_i)}_t \, 
\phi_i \left(t, \, \overline{X}^{(\widehat{\varsigma}+ \epsilon \, \phi_i)}_t \right) \, dW_t \\
& &+  \overline{X}^{(\widehat{\varsigma}+ \epsilon \, \phi_i)}_t \left[ 
\widehat{\varsigma} \left(t, \, \overline{X}^{(\widehat{\varsigma}+ \epsilon \, \phi_i)}_t\right) - 
\widehat{\varsigma} \left(t, \, \overline{X}^{(\widehat{\varsigma})}_t \right)   \right]  \,dW_t  \\
& & + \overline{Y}^{(\widehat{\varsigma}, \, \phi_i, \epsilon)}_t   
\widehat{\varsigma} \left(t, \,\overline{X}^{(\widehat{\varsigma})}_t\right) \, dW_t \, .
\end{eqnarray*}
Assumption  \ref{assumption-math} guarantees enough regularity on $\widehat{\varsigma}$ to expand the
second term. In particular, it exists
$$
\overline{J}^{(\widehat{\varsigma}, \,\phi_i, \, \epsilon)}_t  \in  \left]  \min \left( 
\overline{X}^{(\widehat{\varsigma}+ \epsilon \, \phi_i)}_t,
\, \overline{X}^{(\widehat{\varsigma})}_t\right), \, \max \left( \overline{X}^{(\widehat{\varsigma}+ \epsilon \, \phi_i)}_t,\, 
\overline{X}^{(\widehat{\varsigma})}_t \right) \right[ \, ,
$$
such that    
\begin{eqnarray*}
d Y_t^{(\widehat{\varsigma}, \, \phi_i, \epsilon)} &= & \epsilon \,  
\overline{X}^{(\widehat{\varsigma}+ \epsilon \, \phi_i)}_t \, 
\phi_i \left(t, \, \overline{X}^{(\widehat{\varsigma}+ \epsilon \, \phi_i)}_t \right) \, dW_t \\
& & + \overline{Y}^{(\widehat{\varsigma}, \, \phi_i, \epsilon)}_t   
\left[ \widehat{\varsigma} \left(t, \,\overline{X}^{(\widehat{\varsigma})}_t\right) + 
\overline{X}^{(\widehat{\varsigma}+ \epsilon \, \phi_i)}_t  \, 
\frac{\partial \widehat{\varsigma}}{\partial x} \left(t, \, \overline{J}^{(\widehat{\varsigma}, \, \phi_i, \epsilon)}_t
\right) \right] \, dW_t \, .
\end{eqnarray*}
Moreover, it is possible to choose $\overline{J}^{(\widehat{\varsigma}, \,\phi_i, \, \epsilon)}$ to be path-continuous
 thanks to the regularities of $\widehat{\varsigma}$ and the path-continuity of $\overline{X}$. Then $\overline{J}^{(\widehat{\varsigma}, \,\phi_i, \, \epsilon)}$ is predictable.
 
 We define the following processes:
\begin{eqnarray*}
H_t & = &   \int_0^t  \overline{X}^{(\widehat{\varsigma}+ \epsilon \, \phi_i)}_s \, 
\phi_i \left(s, \, \overline{X}^{(\widehat{\varsigma}+ \epsilon \, \phi_i)}_s \right) \, dW_s \\
Z_t & = & \int_0^t \left[ \widehat{\varsigma} \left(s, \,\overline{X}^{(\widehat{\varsigma})}_s\right) + 
\overline{X}^{(\widehat{\varsigma}+ \epsilon \, \phi_i)}_s  \, 
\frac{\partial \widehat{\varsigma}}{\partial x} \left(s, \, \overline{J}^{(\widehat{\varsigma}, \, \phi_i, \epsilon)}_s
\right) \right] \, dW_s \,.
\end{eqnarray*}
We remark that $H$ is an adapted process and $Z$ is a continuous martingale, then we can apply the variation of constants method (see Protter \cite{bib:Protter} Theorem 53 section V.9) to prove that 
$$
Y_t^{(\widehat{\varsigma}, \, \phi_i, \epsilon)} = \epsilon \,  H_t + \epsilon \, \mathcal{E}(Z)_t \int_0^t \mathcal{E}(Z)_s^{-1}
\, \left( H_s \, dZ_s -H_s \, d[Z,\,Z]_s \right)   \, ,
$$ 
where $\mathcal{E}(Z)$ denotes the Doleans-Dade exponential of $Z$. In particular, we remark that $Y_t^{(\widehat{\varsigma}, \, \phi_i, \epsilon)} \propto \epsilon$, then we can replace 
$Y_t^{(\widehat{\varsigma}, \, \phi_i, \epsilon)}$ into equation (\ref{tech-eq-lemma-1}). Thanks to 
Assumption  \ref{assumption-payoff}, the norm of the second term in equation (\ref{tech-eq-lemma-1}) is controlled by $c \, \epsilon^2$, where $c$ is a constant. 
Finally, thanks to the continuity of all processes, the difference quotient converge when $\epsilon$ goes to zero, then the existence of the limit
(\ref{rel:gateaux-1}) is proved.

To prove the existence of the Gateaux derivative of $\Delta(\widehat{\varsigma}, \,S_t, \,t)$ 
defined by the limit
(\ref{rel:gateaux-2}), we have to study the SDE verified by the first derivative of 
$\Phi(X_t^{(\widehat{\varsigma})})$ with respect to the initial condition $x$. This SDE, the existence and 
the uniqueness of its solution and the closed forms of the solution can be proved adapting the 
proof of Theorem 39 section V.7 in  \cite{bib:Protter} and is omitted. 
Then, the proof of the existence of the limit  (\ref{rel:gateaux-2}) and 
of the second order Gateaux derivatives is similar to 
the previous one. 

\end{proof}

We remark moreover that the integrability of $\Delta$ w.r.t. the asset $S$ and, therefore, the fact that the gain process 
$\displaystyle \int_0^T \Delta (\widehat{\varsigma}, S_t, t) dS_t$ is a $\mathbb{Q}_1$-martingale is guaranteed by    
the existence and the uniqueness of the SDE solution verified by the  first derivative of 
$\Phi(X_t^{(\widehat{\varsigma})})$ with respect to the initial condition $x$ combined with the boundedness of the same first
derivative, see hypothesis  \ref{assumption-payoff}.

The previous result has a direct consequence:
\begin{lemmas}[Gateaux derivatives as continuous maps]\label{Lemma-Gateaux-2}\hfill
\vspace{0.2cm}

The Gateaux derivatives $\displaystyle \frac{\partial C}{ \partial \sigma(\phi_i) }$,  
$\displaystyle \frac{\partial \Delta}{ \partial \sigma(\phi_i) }$,  
$\displaystyle \frac{\partial^2 C}{ \partial \sigma(\phi_i) \partial \sigma(\phi_j)}$ 
and  $\displaystyle \frac{\partial^2 \Delta}{ \partial \sigma(\phi_i) \partial \sigma(\phi_j)}$ as function of estimated volatility 
are continuous maps from the metric space $(L^2_{\sigma}, d^2_{\sigma})$ into the space of squared integrable function 
at the point $\widehat{\varsigma}$.
\end{lemmas}

To prove this lemma, we remark that $L^2_{\sigma}$ is an open set of the space of functions spanning 
by $\{\phi_i\}_{i\in \mathbb{N}}$ equipped with the distance induced by the sequence representation. 
The target space is a Hausdorff space, then we have only to prove 
that the limits of our Gateaux derivatives depending on the function $b(t,x)$ when $b(t,x)$ approaches $\widehat{\varsigma}$  
are the same Gateaux derivatives evaluated on  $\widehat{\varsigma}$. This result can be proved adapting the proof of lemma 
\ref{Lemma-Gateaux} and is omitted.

The key remark about P\&L equation   
(\ref{equation:P-and-L}) is that the price $C$ and the portfolio $\Delta$ depend on the 
 volatility process $\widehat{\varsigma}$ estimated by the trader. This estimated volatility is generally is different from 
market volatility  $\sigma$. Moreover, the estimated volatility is just the result of an estimation rule applied to the data set.
The profit and lost process suffers from the uncertainty related to this estimation, since the estimated volatility 
$\widehat{\varsigma}$ must be replaced by the estimator $\varsigma$.
We concentrate our attention on the law of the P\&L at maturity. In the absence of 
parameter uncertainty, 
the random variable $P\& L(T)$ is equal to zero almost surely and the option can be hedged
exactly. However, the exact hedging strategy cannot be implemented since the option seller 
does not know the parameters values $a_i$ but only the estimated $\widehat{A}_i$. Therefore, $P\&L(T)$ is a random variable and 
we have the following remark.  

\begin{remarque}[Random sources]\hfill
\vspace{0.2cm}

The value of the profit and loss process at maturity is a random variable depending on two 
random sources:

\begin{enumerate}
    \item the stochastic "true" model $(\Omega_1, \, \mathcal{F}^1,\, \mathbb{P}_1)$ since the trader 
    cannot use the correct hedging portfolio.
    \item  the space $(\widetilde{\Omega}, \, \widetilde{\mathcal{F}},\, \widetilde{\mathbb{P}})$, i.e.  the stochastic process $\varsigma$, that depends
    on a random component unrelated to the Brownian motion $B_t$.
\end{enumerate}
\end{remarque}
As a consequence $P\& L(T)$ is a random variable in the product space $(\Omega, \mathcal{F}, \mathbb{P})$.

\begin{remarque}[Role of the historical probability]\hfill
\vspace{0.2cm}

The profit and loss process must be studied under the historical probability
$\mathbb{P} = \mathbb{P}_1 \times \widetilde{\mathbb{P}}$. As a matter of fact, the risk neutral probability 
$\mathbb{Q} =\mathbb{Q}_1  \times \widetilde{\mathbb{P}}$ can be used if and 
only if all contingent claims are attainable. In our case, the option
seller does not know the actual diffusion coefficients of the underlying, so the law of $P\&L$ 
is not degenerate. The main impact is that the drift of the SDE (\ref{SDE-SV}) plays a role and 
that the second term in the $P\&L$ process (\ref{equation:P-and-L}) is not a martingale. 
This fact complicates the computations in this paper. The important role of the drift $\mu$ in 
asset pricing when the market is incomplete has already been emphasized in the literature, 
we mention for instance Karatzas et al. \cite{bib:Karatzas} and Lyons \cite{bib:Lyons}. 
\end{remarque}

The two previous remarks show that the computation of the law of the $P\&L$ process is not
immediate. Moreover, the essential ingredient needed to determine the law of the $P\&L$ 
is the law of $\varsigma$, and the law of this process is difficult to determine since it depends 
on the calibration methodology. We assumed in Assumption \ref{assumption-math}, that the 
option seller can estimate the mean value of the parameters and their variance, i.e. the two first
moments of the law of $\varsigma$. 

Given the knowledge of the two first moments of the law of $\varsigma$, we can at best 
estimate the two first moments of the law of $P\&L(T)$. This remark justifies our use of  the 
error theory using Dirichlet forms, see section \ref{sec:error-theory}. As a consequence,
we can define accurately the law of $P\&L(T)$ with respect to the random source on the 
probability space $\mathbb{P}_1$ but  we can just estimate the two first orders of the 
dependency with respect to $\widetilde{\mathbb{P}}$.
In order to analyze the law of  the $P\&L$ process, it is sufficient to
study the $\mathbb{P}_1$-expectation on a class of regular test functions $h(P \& L(T))$ and
the uncertainty on them using error theory. In practice we will compute the bias and the 
variance of $\mathbb{E}_1[h(P\&L(T))]$, where $\mathbb{E}_1$ denotes the expectation under 
the historical probability $\mathbb{P}_1$.  As a consequence of the dependence on $\varsigma$,
$\mathbb{E}_1[h(P\&L(T))]$ is a random variable on the probability space 
$(\widetilde{\Omega}, \, \widetilde{\mathcal{F}},\, \widetilde{\mathbb{P}})$.

For sake of simplicity, we cut the dependency of $P\& L$ on $T$ throughout the rest of the paper.

 \begin{teorema}[Approximate law of the profit and loss process]\label{theorem-main}\hfill
 \vspace{0.2cm}

 Under assumptions  \ref{assumption-fin}, \ref{assumption-math} and  \ref{assumption-payoff} 
 and for all test functions $h$ belonging to $C^2$ with bounded derivatives,  
 we have the following bias and variance:

\begin{eqnarray}
\mathcal{A}[\mathbb{E}_1[h(P \& L)]]  (\sigma, \widehat{\varsigma})
& = &  \sum_i  \Lambda^{(1)}_i (\sigma,\, \widehat{\varsigma}) \,
\mathcal{A}[A_i](\widehat{A}_i)   \label{A-P-and-L} \\
& & +  \frac{1}{2} \sum_{i,j} \Lambda^{(2)}_{i,j} (\sigma, \,\widehat{\varsigma}) \,   
\Gamma[A_i, A_j](\widehat{A}_i, \widehat{A}_j) \nonumber  \\
    \Gamma \left[\mathbb{E}_1\left[h\left( P\&L  \right)\right]
    \right] (\sigma, \widehat{\varsigma})  & = &   \sum_{i,\, j} \Psi_{i} (\sigma, \, \widehat{\varsigma}) \, 
    \Psi_{j} (\sigma, \, \widehat{\varsigma}) \, \Gamma[A_i, A_j](\widehat{A}_i, \widehat{A}_j) \label{gamma-P-and-L},
\end{eqnarray}
where
\begin{eqnarray*}
\Lambda^{(1)}_i(\sigma,\, \widehat{\varsigma}) & = & \mathbb{E}_1\left[ h^{\prime}(P\& L(T) ) \,\left\{ 
\frac{\partial C}{\partial \sigma(\phi_i)} (\widehat{\varsigma}, x,\, 0)  + 
\int_0^T   \frac{\partial \Delta}{\partial  \sigma(\phi_i)} (\widehat{\varsigma}, S_t, t)  \, dS_t  
  \right\}   \right] \\
\Lambda^{(2)}_i(\sigma,\, \widehat{\varsigma})  & =  &  
\displaystyle    \mathbb{E}_1 \left[ h^{\prime}(P\& L)   
\left\{  \frac{\partial^2 C}{\partial \sigma(\phi_i) \partial \sigma(\phi_j) }  (\widehat{\varsigma}, x, 0)    +
\int_0^T  \frac{\partial^2 \Delta}{\partial \sigma(\phi_i) \partial \sigma(\phi_j)   
}  (\widehat{\varsigma}, S_t,\, t) \; dS_t \right\} \right] \\ 
   & & +  \displaystyle 
 \mathbb{E}_1 \left[ h^{\prime \prime}(P\& L) \; 
 \left\{ \frac{\partial C}{\partial \sigma(\phi_i)}(\widehat{\varsigma},\, x, \, 0) + 
  \int_0^T  \frac{\partial \Delta}{\partial \sigma(\phi_i)}(\widehat{\varsigma}, \, S_s, \,s)  \,  
  dS_s         \right\}  \right. \times  \\
  & &\quad \quad  \quad  \times \displaystyle \left.  \left\{ \frac{\partial C}{\partial \sigma(\phi_j)}
  (\widehat{\varsigma},\, x, \, 0) + 
  \int_0^T  \frac{\partial \Delta}{\partial \sigma(\phi_j)}(\widehat{\varsigma}, \, S_s, \,s)  \,  
  dS_s         \right\}     \right]  \\
\Psi_i(\sigma,\, \widehat{\varsigma})  & = &  \mathbb{E}_1 \left[ h^{\prime}(P\& L)
\, \left\{\frac{\partial C}{\partial \sigma (\phi_i)}(\widehat{\varsigma}, x,
    0)  + \int_0^T  \frac{\partial \Delta}{\partial \sigma(\phi_i)}(\widehat{\varsigma}, \, S_s, s) \, dS_s  
    \right\}  \right]  
\end{eqnarray*}

\end{teorema}

\begin{proof}
For sake of simplicity, we neglect the dependency in $T$ of  $P\& L$. 
The proof is split into two parts, we first prove the relation (\ref{gamma-P-and-L}) and 
then (\ref{A-P-and-L}).
Thanks to Assumption \ref{assumption-math}, the operator $\Gamma$ admits a sharp operator. 
In particular,  it exists a sharp operator denoted by $(\cdot)_i^{\#}$ on each sub-error structure associated to each 
random variable $A_i$ and the sharp operator $(\cdot)^{\#}$ on the global error structure is the sum of the sharp 
operators in each sub-error structure.

We study the sharp of $\mathbb{E}_1\left[h\left( P\&L \right)\right]$. Thanks to the linearity of the 
sharp operator and the smoothness of the test function $h$ we have
\begin{equation}\label{eqn:proof-gamma}
\begin{array}{rcl}
 \displaystyle \left(\mathbb{E}_1\left[h\left( P\&L  \right)\right]\right)^{\#} & = & 
  \displaystyle   \mathbb{E}_1\left[h'\left(P\&L \right) \; (P\&L)^{\#} \right]  \\
    & =  & \displaystyle \mathbb{E}_1\left[h'\left(P\&L \right) \;
    \left( C(\varsigma, \, x, \, 0) + \int_0^T \Delta(\varsigma, \, S_t, \, t) dS_t - \Phi(S_T) \right)^{\#}\right] \, .
    \end{array}
\end{equation}
Using the linearity of the sharp, the term into braces can be rewritten as
$$
 \left(C(\varsigma, \, x, \, 0)  \right)^{\#}  + \left( \int_0^T \Delta(\varsigma, \, S_t, \, t) dS_t  \right)^{\#}  - \left( \Phi(S_T) \right)^{\#} 
$$
 are  we remark that $(\Phi(S_T))^{\#}=0$ since the payoff is independent on the volatility 
estimated by the option seller. 
The integral is a linear operator defined by a $L^2$-limit, we proceed by approximation replacing the integral by a sum.
We recall that the "true" diffusion $S$ does not suffer from the uncertainty on estimated volatility. 
Using the linearity and the closedness of the sharp operator, we can then write
$$
\left( \int_0^T \Delta(\varsigma, \, S_t, \, t) \, dS_t\right)^{\#} =  \int_0^T \left(\Delta(\varsigma, \, S_t, \, t) \right)^{\#} \,dS_t,
$$
where we have used the first point of Assumption \ref{assumption-fin}, see proposition V.8 page 83 in 
\cite{bib:Bouleau-erreur} for a more detailed analysis.
Thanks to the expansion of the volatility (see Assumption \ref{assumption-math}), the linearity of 
the sharp operator (see Proposition \ref{prop:sharp}), we have     
\begin{equation}\label{eqn:proof-sharp}
\begin{array}{rcl}
\displaystyle C(\varsigma,\, x, \, 0)^{\#} &= & \displaystyle  \sum_i  \frac{\partial C}{\partial \sigma(\phi_i)} 
(\varsigma, x,\, 0)  \; A_i^{\#} \\
\displaystyle  \left(\int_0^T \Delta(\varsigma, \, S_t, \, t) \, dS_t\right)^{\#} &= &
\displaystyle \sum_i A_i^{\#} \; \int_0^T
 \frac{\partial \Delta}{\partial \sigma(\phi_i)}(\varsigma, \, S_t, \, t) \;  dS_t , 
\end{array}
\end{equation}
We  multiply by $h^{\prime}(P\& L)$ and we take the $\mathbb{P}_1$-expectation 
of the two right sides. Thanks to Assumption \ref{assumption-math}, $A_i^{\#}$ is defined on a probability space $(\widetilde{\Omega} \times \overline{\Omega})$ distinct from $(\Omega_1, \, \mathcal{F}^1, \, \mathbb{P}_1)$, then we have
\begin{eqnarray*}
\mathbb{E}\left[ h^{\prime}(P\& L) \, C(\varsigma,\, x, \, 0)^{\#} \right] &= & 
\sum_i  \mathbb{E}\left[ h^{\prime}(P\& L) \,   \frac{\partial C}{\partial \sigma(\phi_i)} 
(\varsigma, x,\, 0) \right]  \; A_i^{\#} \\
\mathbb{E}\left[ h^{\prime}(P\& L) \,  \left(\int_0^T \Delta(\varsigma, \, S_t, \, t) \, dS_t\right)^{\#} \right] 
&= & \sum_i \mathbb{E}\left[ h^{\prime}(P\& L)  \int_0^T 
\frac{\partial \Delta}{\partial \sigma(\phi_i)}(\varsigma, \, S_t, \, t) \, dS_t  \right] \, A_i^{\#}  \, 
\end{eqnarray*}
then
$$
\left(\mathbb{E}_1\left[h\left( P\&L  \right)\right]\right)^{\#} =  \sum_i \Psi_i(\sigma,\, \varsigma) A_i^{\#}\, .
$$
We conclude the proof on the relation for the quadratic error 
$\Gamma\left[ \; \mathbb{E}\left[h\left( P\&L \right)\right]\;  \right]$  using  the first property of the sharp operator 
(see Proposition \ref{prop:sharp}), i.e.
\begin{eqnarray*}
\Gamma\left[ \; \mathbb{E}_1\left[h\left( P\&L \right)\right]\;  \right] & = & \overline{\mathbb{E}}\left[
\left\{\left(\mathbb{E}_1\left[h\left( P\&L  \right)\right]\right)^{\#}\right\}^2 \right] = \sum_{i,j} 
\Psi_i(\sigma,\, \varsigma) \Psi_j(\sigma,\, \varsigma) \overline{\mathbb{E}}\left[
A_i^{\#} A_j^{\#}\right]  \\ 
& = &  \sum_{i,j} 
\Psi_i(\sigma,\, \varsigma) \Psi_j(\sigma,\, \varsigma) \Gamma[A_i, A_j] 
\end{eqnarray*}
and we find finally the relation (\ref{gamma-P-and-L}) evaluating the last random variable 
on the particular realization corresponding to observed data.

The study of the bias is more complex, we first apply the closedness of  the bias operator and the
chain rule (\ref{bias-chain-rule}), then we find
\begin{equation}\label{eqn:proof-bias}
 \mathcal{A} \left[\mathbb{E}_1\left[h(P \& L)\right]\right] = 
 \mathbb{E}_1\left[ \mathcal{A}[h(P \& L)]\right] =
  \mathbb{E}_1 \left[ h'(P \& L) \, \mathcal{A}\left[ P \& L  \right] + \frac{1}{2} h^{\prime
    \prime}(P \& L)  \, \Gamma \left[P \& L\right]\right] ,
  \end{equation}
We study the two terms separately. 
We first consider the term depending on the variance. 
Thanks to the previous analysis, in particular relations (\ref{eqn:proof-sharp}), and the properties of  the sharp operator,  we have 
\begin{eqnarray*}
\Gamma  \left[P \& L \right] = \sum_{i,j}   \left\{ \frac{\partial C}{\partial \sigma(\phi_i)} 
(\varsigma, x,\, 0)   + \int_0^T
 \frac{\partial \Delta}{\partial \sigma(\phi_i)}(\varsigma, \, S_t, \, t) \;  dS_t
\right\}^2 \; \Gamma[A_i,A_j] \, .
 \end{eqnarray*}
We multiply by $h^{\prime \prime}(P\& L)$, we take the expectation 
under $\mathbb{P}_1$ and we find the third term in (\ref{A-P-and-L}).
We now study $\mathcal{A}[P\&L]$, we apply the linearity of the bias operator and we find
$$
\mathcal{A}[P\&L] = \mathcal{A}\left[C(\varsigma,\, x, \,0)\right] + \mathcal{A}\left[ \int_0^T \Delta(\varsigma,\, S_t, \,t)\, dS_t   \right] -  \mathcal{A}\left[ \Phi(S_T)\right].
$$
The last term is worth zero, since the final payoff is completely defined by $S_T$ and does 
not depend on the volatility estimated by the option seller, see Assumptions 
\ref{assumption-fin} and \ref{assumption-payoff}. 
Thanks to the same assumptions and the closedness of the bias operator $\mathcal{A}$, 
we have
$$
\mathcal{A}\left[ \int_0^T \Delta(\varsigma,\, S_t, \,t)\, dS_t   \right] =  \int_0^T 
\mathcal{A}\left[\Delta(\varsigma,\, S_t, \,t) \right]\, dS_t.
$$
Thanks to the same argument used for the study of the variance, we can take the 
Gateaux-derivatives of $C$ and $\Delta$ with respect to a variation of the 
volatility along the component $\phi_i$ and using the chain rule 
(\ref{bias-chain-rule}), we have  
\begin{equation*}
\begin{array}{rcl}
\displaystyle \mathcal{A}\left[C(\varsigma,\, x, \,0)\right]   & = & \displaystyle
 \sum_i  \frac{\partial C}{\partial \sigma(\phi_i)} 
(\varsigma, x,\, 0)  \; \mathcal{A}[A_i] \,  \\
& & \displaystyle + \frac{1}{2}  \sum_{i,j} \frac{\partial^2 C}{\partial \sigma(\phi_i) \partial \sigma(\phi_j)  
}  (\varsigma, x,\, 0) \;  \Gamma[A_i,A_j] \\
\displaystyle  \int_0^T  \mathcal{A}\left[\Delta(\varsigma,\, S_t, \,t) \right]\, dS_t  
& = & \displaystyle \sum_i \left[ \int_0^T \frac{\partial \Delta}{\partial \sigma(\phi_i)} 
(\varsigma, S_t,\, t) \;  dS_t \right]\;  \mathcal{A}[A_i]   \\
 &  &  \displaystyle  + \frac{1}{2}  \sum_{i,j}  \left[\int_0^T 
\frac{\partial^2 \Delta}{\partial  \sigma(\phi_i) \partial  \sigma(\phi_j) 
}  (\varsigma, S_t,\, t) \;  dS_t \right] \;  \Gamma[A_i,A_j]  \, . 
\end{array}
\end{equation*}
We multiply by $h^{\prime}(P\& L)$, we take the $\mathbb{P}_1$-expectation and we find the two first terms in
(\ref{A-P-and-L}) evaluating this random variable
on the particular realization corresponding to observed data.

thanks to the continuity at the point $(\widehat{A}_1, \ldots
\widehat{A}_k, \ldots)$.

\end{proof}

We remark that the approximate law of the profit and loss process depends both on the true 
volatility $\sigma(t,\, x)$ and on the estimated one $\widehat{\varsigma}(t,\, x)$.
Two natural problems arises. First of all, the true volatility $\sigma$ is unknown. Moreover, 
the volatility estimator $\varsigma$ is a random variable, hence the bias and the variance 
of  $\mathbb{E}_1[h(P\& L)]$, given by Theorem \ref{theorem-main}, are conditional
moments known given the expected volatility $\widehat{\varsigma}$. It is important to remark that   
$\mathbb{E}_1[h(P\& L)]$ remains a random variable on the probability space $\widetilde{\Omega}$.
We explicit this fact indicating the dependency on $\varsigma$ throughout the rest of the paper.  

In a particular case, we can partially overcome the problem of the dependence on $\sigma$.
We have the following corollary 

 \begin{corollario}[Bias and Variance of the expected value 
 of $P\& L$]\label{theorem-corollary}\hfill
 \vspace{0.2cm}

 Under assumptions  \ref{assumption-fin}, \ref{assumption-math} and  \ref{assumption-payoff},
 we have the following bias and variance:

\begin{eqnarray}
\mathcal{A}[\mathbb{E}_1[P \& L(T)]](\sigma, \widehat{\varsigma})  & = &  \sum_i \overline{\Lambda}^{(1)}_i 
(\sigma,\, \widehat{\varsigma}) \, \mathcal{A}[A_i](\widehat{A}_i)    \label{A-P-and-L-2} \\
& & + \frac{1}{2} \sum_{i,j} \overline{\Lambda}^{(2)}_{i,j}
(\sigma, \,\widehat{\varsigma}) \,   \Gamma[A_i,A_j](\widehat{A}_i, \widehat{A}_j) \nonumber \\
    \Gamma \left[\mathbb{E}_1\left[ P\&L\right] \right](\sigma, \widehat{\varsigma}) & = &   
    \sum_{i,j}  \overline{\Psi}_{i} (\sigma, \, \widehat{\varsigma})  \overline{\Psi}_{j} (\sigma, \, \widehat{\varsigma}) 
    \, \Gamma[A_i,A_j] (\widehat{A}_i, \widehat{A}_j)  \label{gamma-P-and-L-2},
\end{eqnarray}
where
\begin{eqnarray*}
\overline{\Lambda}^{(1)}_i(\sigma,\, \widehat{\varsigma}) & = & 
\frac{\partial C}{\partial \sigma(\phi_i)} (\widehat{\varsigma},\, x,\, 0)  + \mu
\int_0^T  \mathbb{E}\left[ \frac{\partial \Delta}{\partial  \sigma(\phi_i)} 
(\widehat{\varsigma}, \, S_t,\, t)  \, S_t \right] \, dt  \\
\overline{\Lambda}^{(2)}_{i,j}(\sigma,\, \widehat{\varsigma})  & =  &  
\displaystyle    \frac{\partial^2 C}{\partial \sigma(\phi_i) \partial \sigma(\phi_j) }  (\widehat{\varsigma}, \, x,\, 0)
 \, \mu \int_0^T \mathbb{E} \left[  
\frac{\partial^2 \Delta}{\partial [\sigma(\phi_i)]^2}  (\widehat{\varsigma}, \, S_t,\, t) \, S_t  \right] \, dt \\ 
\overline{\Psi}_i(\sigma,\, \widehat{\varsigma})  & = & 
\frac{\partial C}{\partial \sigma (\phi_i)}(\widehat{\varsigma},\, x, \,  0) + \mu \int_0^T
\mathbb{E}_1 \left[   \frac{\partial \Delta}{\partial \sigma(\phi_i)}(\widehat{\varsigma}, \, S_t,\, t) \, S_t 
  \right] \, dt     \, .
\end{eqnarray*}

\end{corollario}

\begin{proof}
The proof is a direct consequence of the Theorem \ref{theorem-main}, where we have calculate
the expectation using the SDE (\ref{SDE-SV}), we omit the details. 
\end{proof}

In particular we remark that $\overline{\Lambda}^{(1)}_i(\sigma,\, \widehat{\varsigma})$, 
$\overline{\Lambda}^{(2)}_i(\sigma,\, \widehat{\varsigma})$ and $\overline{\Psi}_i(\sigma,\, \widehat{\varsigma})$
depend on $\sigma$ by means of the portfolio term. Moreover, assuming $\mu=0$, the 
dependence on $\sigma$ vanishes. In other words, the bias and the variance of the 
expected value of the option seller wealth are independent on the true volatility in a market 
without risk premium. 

In the general case, we can approximate the equations (\ref{A-P-and-L}) and 
(\ref{gamma-P-and-L}) using the known value $\widehat{\varsigma}$ instead of the unknown 
$\sigma$. Under 
this approximation, we have the following corollary.

\begin{corollario}[Approximated tails of the law of the $P\& L$]\hfill
 \vspace{0.2cm}

Under the hypotheses of Theorem \ref{theorem-main} and assuming the approximation 
$\sigma = \widehat{\varsigma}$, the random variable $ \mathbb{E}_1\left[h(P\&L)\right](\varsigma)$, defined on probability space
$(\widetilde{\Omega}, \widetilde{\mathcal{F}}, \widetilde{\mathbb{P}})$, has the following approximated tails:  

\begin{equation}\label{eq:cheb}
\widetilde{\mathbb{P}}\left[  \mathbb{E}_1\left[h(P\&L)\right] (\varsigma) - 
h(0)-   \mathcal{A} \left[  \mathbb{E}_1\left[h(P\&L)\right]\right](\widehat{\varsigma})   \geq k \,  \sqrt{
   \Gamma \left[\mathbb{E}_1\left[h\left( P\&L \right)\right]
    \right]   (\widehat{\varsigma}) }   \right] \leq \frac{1}{1+k^2} \, , 
\end{equation}
where 
\begin{eqnarray}
\mathcal{A}[\mathbb{E}_1[h(P \& L)]](\widehat{\varsigma}) & = &  \sum_i  \Lambda^{(1)}_i 
(\widehat{\varsigma},\, \widehat{\varsigma}) \,
\mathcal{A}[A_i] (\widehat{A}_i)  \label{A-P-and-L-approx} \\
& & +\frac{1}{2}   \sum_{i,j} \Lambda^{(2)}_{i,j} (\widehat{\varsigma}, \,\widehat{\varsigma}) \,   \Gamma[A_i,A_j] 
(\widehat{A}_i, \widehat{A}_j)  \nonumber \\
    \Gamma \left[\mathbb{E}_1\left[h\left( P\&L \right)\right]  \right]  (\widehat{\varsigma})  
    & = &   \sum_{i,j}  \Psi_{i} (\widehat{\varsigma}, \, \widehat{\varsigma})
    \Psi_{j} (\widehat{\varsigma}, \, \widehat{\varsigma}) \; \Gamma[A_i,A_j](\widehat{A}_i, \widehat{A}_j)
     \label{gamma-P-and-L-approx},
\end{eqnarray}
with
\begin{eqnarray*}
\Lambda^{(1)}_i(\widehat{\varsigma},\, \widehat{\varsigma}) & = &  h^{\prime}(0) \,\left\{ 
\frac{\partial C}{\partial \sigma(\phi_i)} (\widehat{\varsigma}, x,\, 0)  + \mu 
\int_0^T  \mathbb{E}_1\left[   X_t^{(\widehat{\varsigma})} \,
\frac{\partial \Delta}{\partial  \sigma(\phi_i)} (\widehat{\varsigma}, X^{(\widehat{\varsigma})}_t,\, t)
 \right]  dt \right\}   \\
\Lambda^{(2)}_{i,j} (\widehat{\varsigma},\, \widehat{\varsigma})  & =  &  
\displaystyle    h^{\prime}(0)   
\left\{  \frac{\partial^2 C}{\partial \sigma(\phi_i) \partial \sigma(\phi_j)}  (\widehat{\varsigma}, x,\, 0)    + \mu
\int_0^T  \mathbb{E}_1 \left[  X^{(\widehat{\varsigma})}_t \, 
\frac{\partial^2 \Delta}{\partial \sigma(\phi_i) \partial \sigma(\phi_j)}  (\widehat{\varsigma}, X^{(\widehat{\varsigma})}_t,\, t) 
\right]\, dt \right\}  \\ 
   & & +  \displaystyle 
h^{\prime \prime}(0) \; \left\{ \frac{\partial C}{\partial \sigma(\phi_i)}(\widehat{\varsigma},\, x, \, 0) + \mu
 \int_0^T  \mathbb{E}_1 \left[ X^{(\widehat{\varsigma})}_s \,  \frac{\partial \Delta}{\partial \sigma(\phi_i)}
(\widehat{\varsigma}, \, X^{(\widehat{\varsigma})}_s, \,s) \right] \,  ds         \right\}   \\
& &\quad \quad \quad  \times  \displaystyle  \left\{ \frac{\partial C}{\partial \sigma(\phi_j)}(\widehat{\varsigma},\, x, \, 0) + \mu
 \int_0^T  \mathbb{E}_1 \left[ X^{(\widehat{\varsigma})}_s \,  \frac{\partial \Delta}{\partial \sigma(\phi_j)}
(\widehat{\varsigma}, \, X^{(\widehat{\varsigma})}_s, \,s) \right] \,  ds         \right\}    \\
& & + \displaystyle  h^{\prime \prime}(0)   \int_0^T  \mathbb{E}_1 
\left[  \widehat{\varsigma}^2\,  \left\{X^{(\widehat{\varsigma})}_s\right\}^2  \,  \frac{\partial \Delta}{\partial \sigma(\phi_i)} 
(\widehat{\varsigma}, \, X^{(\widehat{\varsigma})}_s, \,s)   \frac{\partial \Delta}{\partial \sigma(\phi_j)} 
(\widehat{\varsigma}, \, X^{(\widehat{\varsigma})}_s, \,s)  \right] \,  ds \\
\Psi_i(\widehat{\varsigma},\, \widehat{\varsigma})  & = &   h^{\prime}(0)
\frac{\partial C}{\partial \sigma (\phi_i)}(\widehat{\varsigma}, x, \, 0)  
+ \mu \int_0^T \mathbb{E}_1 \left[X^{(\widehat{\varsigma})}_s \, 
\frac{\partial \Delta}{\partial \sigma(\phi_i)}(\widehat{\varsigma}, \, X^{(\widehat{\varsigma})}_s,\, s)  
 \right] ds     \, .
\end{eqnarray*}

The right-hand member in relation (\ref{eq:cheb}) can be replaced by the complementary error function assuming the Gaussian approximation.
\end{corollario}

\begin{proof}
We first make an useful observation:  In a complete market without uncertainty, the price C and the
portfolio $\Delta$ are defined in such way that the $P\& L$ process is worth zero almost surely.
The changing of $\sigma$ for the estimated value $\widehat{\varsigma}$ into the definitions of 
$\Lambda^{(j)}_i$ and $\Psi_{(i)}$  entail that the $P\& L$ is substituted by the approximated
process
$$
\overline{P\& L} =  C(\widehat{\varsigma}, \, x, \, 0) + \int_0^T \Delta(\widehat{\varsigma}, \, X^{(\widehat{\varsigma}
)}_t, \, t)  \, dX^{\widehat{\varsigma}}_t - \Phi(X^{(\widehat{\varsigma})}_T).
$$ 
Hence, by definition of $C$ and $\Delta$, we have that $\overline{P\& L}= 0$ $\mathbb{P}_1$-a.s..   
We know the  bias and the variance of $ \mathbb{E}_1\left[h(P\&L)\right]$ thanks to 
Theorem  \ref{theorem-main}.  Then, we approximate them using $\sigma=\widehat{\varsigma}$ 
and we find the approximated bias $\mathcal{A}[\mathbb{E}[h(P \& L)]](\widehat{\varsigma}) $ and the
approximated variance  $\Gamma \left[\mathbb{E}\left[h\left( P\&L \right)\right]
    \right](\widehat{\varsigma})$ thanks to an easy calculation using the SDE (\ref{SDE:perturbed}).
Finally, we apply Chebyshev's inequality, see  Corollary
\ref{prop:Cheb}, and we find the relation (\ref{eq:cheb}) thanks to the previous remark.

\end{proof}

The last corollary gives us the approximated tails of the law of the $P\& L$, in particular it 
permits to estimated the probability of a rare event thanks to formula (\ref{eq:cheb}). 
In the next section, we will use these results in order to define the price of the contingent claim.

\section{Option Pricing}\label{sec:OP}

In order to interpret this result in financial terms, we consider that the option seller 
is aware of the presence of errors in his procedure for estimating the volatility $\sigma$ 
and wants to take this into account.
Since the option seller does not control the errors,  
the risk related to the space $(\widetilde{\Omega}, \, \widetilde{\mathcal{F}}, 
\, \widetilde{\mathbb{P}})$, i.e. to the uncertainty on the parameters, cannot be hedged.
Indeed, if we compute the super-hedging price of a contingent claim, we find a  very high buy-price and a very small sell-price, i.e. the bid-ask spread becomes too large compared with the market spreads.  
We also remark that the probability that the random variable $\mathbb{E}[(P\&L)]$ 
takes values far from the mean is very small and becomes negligible if this distance is big compared with the standard deviation $\sqrt{ \Gamma [\mathbb{E}[(P\&L)]] }$.
Therefore, we introduce the following principle to define the price of a contingent claim under this uncertainty model.

\begin{principle}[Asset pricing principle under uncertainty]\label{principle-1}\hfill
\vspace{0.2cm}

The option seller  fixes a risk tolerance $\alpha < 0.5$ and
accepts to sell the option at any price $F_{\text{sell}}(\varsigma)$, 
$\mathcal{F}_0$-measurable, such that 
\begin{equation}\label{eqn-principe-1}
\widetilde{\mathbb{P}}\left\{  \, \mathbb{E}_1\left[  F_{\text{sell}}(\widehat{\varsigma}) 
- C(\widehat{\varsigma})
+  P \& L(\varsigma)\right]  < 0 \right\} \leq \alpha
\end{equation}
where $C(\widehat{\varsigma})$ denotes the cost of the hedging strategy, that is the theoretical price 
of the option without uncertainty on the volatility, i.e. 
$C(\widehat{\varsigma}) = \mathbb{E}^{\mathbb{Q}_1^{\widehat{\varsigma}}}
[\Phi(X^{(\widehat{\varsigma})}_T)]$. 
\end{principle}

Before applying this principle to our analysis, we discuss its financial implications.
First of all, we remark that $F_{\text{sell}}(\widehat{\varsigma})$ and $C(\widehat{\varsigma}) $ 
are $\mathcal{F}_0$ measurable, then the relation 
(\ref{eqn-principe-1}) can be written as
\begin{equation}\label{eqn-principe-2}
\widetilde{\mathbb{P}}\left\{  F_{\text{sell}}(\widehat{\varsigma}) - C(\widehat{\varsigma})
+ \mathbb{E}_1\left[   P \& L(\varsigma)\right]  
< 0 \right\} \leq \alpha \, . 
\end{equation} 
Moreover, the price $F_{\text{sell}}(\widehat{\varsigma})$ depends on the estimated volatility $\widehat{\varsigma}$.
The second remark is that this principle does not give a single price but a half-line of possible selling prices. Indeed if $X$ is a selling price and $Y>X$ then 
$$
\widetilde{\mathbb{P}}\left\{  Y - C(\varsigma)+ \mathbb{E}_1
\left[P \& L(\varsigma)\right]  < 0 \right\} < \widetilde{\mathbb{P}}
\left\{ X - C(\widehat{\varsigma})+ \mathbb{E}_1\left[P \& L(\varsigma)\right]  < 0 \right\} \leq \alpha,
$$
so $Y$ is also a selling price.
We call ask price, denoted $F_{\text{ask}}(\widehat{\varsigma})$, the infimum of the set of all 
acceptable selling prices. 
We also remark that the classical option price in the model without uncertainty is a sell price 
and, in particular, is the ask price, since the wealth 
of an option seller  is worth zero  and the probability space $\widetilde{\mathbb{P}}$ is trivial, 
in a complete market without uncertainty.

We discuss the event $\{ \mathbb{E}_1\left[ X - C(\widehat{\varsigma})+ P \& L(\varsigma)\right]  < 0\}$. 
On this event, if the option seller has sold the option at the price $X$,  then he will  lose 
money in mean at maturity, since he has to pay $C(\widehat{\varsigma})$ to buy the hedging portfolio 
and the noise on this hedging strategy conducts him to lose money in mean. However,  
the probability of this event is smaller than $\alpha$. 

Therefore, the asset pricing principle says that the option seller accepts to sell the option at a 
price $X$ if this price is high enough to guarantee that he loses money with a probability 
smaller than $\alpha$. The choice of the parameter $\alpha$ depends on the risk aversion of 
the option seller. However, if he is too risk-adverse, then he will advertise  the contingent claim
at  too high a price, and potential buyers will find other traders that offer the same option at a 
lower price. 

We remark that there is a likeness between this principle and the VaR of the trader portfolio. 
Indeed, $\alpha$ can be interpreted as the probability  to "lose money on the contract", 
which is the definition of the Value-at-Risk.  
We also remark that  the final wealth of the trader is negative in mean on the event 
$\{\mathbb{E}_1\left[X - C(\widehat{\varsigma})+ P \& L(\varsigma)\right]  < 0\}$; 
this event is the most onerous then we name the related risk  the "risk of the first kind". 
The "risk of the second kind" for the trader is to propose a price too high 
such that potential buyers do not buy the option, i.e. "lose the  opportunity to make money" with
the contract. 
In a future paper, we will study the problem of the optimal proposed price as an equilibrium
between the two previous kinds of risk.  We assume here that the  risk tolerance $\alpha$ is 
fixed and that the option seller chooses the smallest price consistent with his  risk tolerance.

We also remark that the Principle \ref{principle-1} also defines the purchase price. 
As a matter of fact, if a trader accepts to buy an option, he has to take a negative position 
on the hedging portfolio to cover it, i.e. he has to follow the opposite of the hedging strategy. 
In accord with the Principle \ref{principle-1}, the option buyer accepts to buy the option at any 
price $F_{\text{buy}}(\widehat{\varsigma})$ such that
$$
\widetilde{\mathbb{P}}\left\{ - F_{\text{buy}}(\widehat{\varsigma}) + C(\widehat{\varsigma})- 
\mathbb{E}_1\left[ P \& L (\varsigma)\right]  < 0 \right\} \leq \alpha,
$$
that is the relation (\ref{eqn-principe-1}), where we have changed all signs. The 
relation above can be rewritten as 
$$
\widetilde{\mathbb{P}}\left\{ F_{\text{buy}}(\widehat{\varsigma}) - C(\widehat{\varsigma}) 
+ \mathbb{E}_1\left[ P \& L (\varsigma)\right]  > 0 \right\} \leq \alpha.
$$
The previous relation and the relation (\ref{eqn-principe-1}) lead us to the following remark:

\begin{remarque}[Bid-Ask Spread]\hfill
\vspace{0.2cm}

If all option traders on the market are risk adverse, then the best purchase price, 
denoted by $F_{\text{bid}}(\widehat{\varsigma})$, and the best seller price 
$F_{\text{ask}}(\widehat{\varsigma})$ are distinct.
\end{remarque} 
This remark results from the fact that if the traders are risk adverse, then the  risk tolerance 
must be smaller that $0.5$, so $F_{\text{bid}}(\widehat{\varsigma}) < C(\widehat{\varsigma}) + 
\widetilde{\mathbb{E}}  \left[   \mathbb{E}_1[P\&L(\varsigma)] \right] < 
F_{\text{ask}}(\widehat{\varsigma})$. 

The previous principle and Theorem \ref{theorem-main} have the following immediate consequence:

\begin{proposizione}[Option prices]\label{PROP:OP}\hfill
\vspace{0.2cm}

We assume the hypotheses of Theorem  \ref{theorem-main}, the approximation 
$\sigma=\widehat{\varsigma}$ and 
that the option seller follows the principle \ref{principle-1}.
Then he accepts to sell the option at any price higher than 
\begin{displaymath}
F_\text{ask}(\widehat{\varsigma}) = C(\widehat{\varsigma}) +  \mathcal{A}
\left[\mathbb{E}_1[P \& L] \right] (\widehat{\varsigma}) + \sqrt{  \Gamma
\left[\mathbb{E}_1[P \& L]\right] (\widehat{\varsigma}) }\; \mathcal{N}_{1-\alpha}
\end{displaymath}
where $\mathcal{N}_{1-\alpha}$ is the $(1-\alpha)$-quantile of the reduced
normal law, or the function $\sqrt{\frac{\alpha}{1-\alpha}}$ given by Chebychev's 
inequality (\ref{cheb-equation}) in the conservative case. Likewise, the option buyer 
accepts to buy the option at any
price lower than
\begin{displaymath}
F_\text{bid}(\widehat{\varsigma}) = C(\widehat{\varsigma}) +  \mathcal{A}
\left[\mathbb{E}_1[P \& L]\right](\widehat{\varsigma}) + \sqrt{ \Gamma
\left[\mathbb{E}_1[P \& L]\right](\widehat{\varsigma})} \; \mathcal{N}_{\alpha}
\end{displaymath}

\end{proposizione}

We remark some symmetry in the two previous prices, since
$\mathcal{N}_{\alpha} + \mathcal{N}_{1- \alpha}= 0$. Therefore,  we have
\begin{eqnarray}
F_{\text{mid}}(\widehat{\varsigma}) = \frac{F_{\text{ask}}(\widehat{\varsigma})+F_{\text{bid}}(\widehat{\varsigma})}{2} 
& = & C(\widehat{\varsigma}) +  \mathcal{A} \left[\mathbb{E}_1[P \& L]\right](\widehat{\varsigma}) \\
\text{Spread}(\widehat{\varsigma}) = 
F_{\text{ask}}(\widehat{\varsigma})-F_{\text{bid}}(\widehat{\varsigma})& = & 2 \sqrt{ \Gamma
\left[\mathbb{E}_1[P \& L]\right] (\widehat{\varsigma})} \, \mathcal{N}_{1-\alpha}
\end{eqnarray}
where $F_{\text{mid}}(\widehat{\varsigma})$ denotes the mid price, i.e. the average of the bid and 
ask prices, and $\text{Spread}(\widehat{\varsigma})$ denotes the bid-ask spread, i.e. the difference
between the ask and the bid prices.    
We emphasize that with our model, we can reproduce a bid-ask spread
and we can associate its width to the trader's risk aversion, i.e. the
probability $\alpha$, and the volatility uncertainty, i.e. the term
$\sqrt{ \Gamma \left[\mathbb{E}_1[P \& L]\right](\widehat{\varsigma})}$.
Another interesting point it that the mid-price does not depend on 
the  risk tolerance $\alpha$.

We conclude this section with a remark on the Assumption \ref{assumption-payoff}.

\begin{remarque}[General Payoff]\label{REM:GP}\hfill
\vspace{0.2cm}

It is clear that the Assumption \ref{assumption-payoff} is quite restrictive, 
since call and put options do not verify it, for instance. 
However, we use the hypothesis that the payoff is regular only to assure the existence of the
Gateaux derivatives of the price and the hedging portfolio, see equations 
(\ref{rel:gateaux-1}), (\ref{rel:gateaux-2}) and Lemma \ref{Lemma-Gateaux};
if we can prove that the Gateaux derivatives exist, then Theorem \ref{theorem-main} 
remains valid and we can define the price for an option according to the Principle \ref{principle-1} 
even if the payoff is not regular.
\end{remarque}

\section{Example: log-normal diffusion}

In this section, we give an example of the previous results. In particular,
we consider the log-normal diffusion, i.e. the Black Scholes model.
The underlying follows the SDE  (\ref{Black-Scholes}). In this case, the volatility 
is a real number, so we replace it by the same number multiplied by the constant function. 
We concentrate our analysis on call options. It is clear that a call option does not verify 
the Assumption \ref{assumption-payoff}, but we have a closed form for all greeks in the 
case of a log-normal diffusion and we can check that the Black Scholes pricing formula 
is $C^2$ at each time $t$ strictly smaller than the maturity T and the vega, i.e. the derivative 
w.r.t. the volatility, vanishes when t goes to $T$. These properties guarantee that  Theorem 
\ref{theorem-main} remains valid  
even without Assumption  \ref{assumption-payoff}, see Remark \ref{REM:GP}. 
We assume that the drift $\mu$ vanishes under the historical probability measure $\mathbb{P}_1$,
in order to simplify our numerical computation and to overcome the problem of the dependence 
on $\sigma_0$, see section \ref{sec:PandL}, in particular Corollary \ref{theorem-corollary}. 
We denote by $\varsigma_0$ (resp. $\widehat{\varsigma}_0$) the volatility estimator (resp. the volatility estimated 
by the option seller). For that we consider an error structure of type Ornstein-Uhlenbeck, see Definition \ref{structureOU}.

In this case, the Theorem \ref{theorem-main} has the following corollary.

\begin{corollario}[Bid and Ask prices with log-normal diffusion]\label{pricing-BS}\hfill
\vspace{0.2cm}

If the underlying follows the Black Scholes SDE   (\ref{Black-Scholes}) without drift, then the bias and the variance of the profit and loss process that hedges a call option verify the following equations:

\begin{eqnarray}
\mathcal{A}[C(x, \, K, \, T,\, \varsigma_0)](\widehat{\varsigma}_0) & = & x \; \frac{e^{-\frac{1}{2} d_1^2}}{\sqrt{2
\pi}} \,\sqrt{T} \left\{ \mathcal{A} \left[\varsigma_0\right](\widehat{\varsigma}_0)  +
 \frac{d_1 \, d_2}{2 \, \sigma_0} \, \ \Gamma\left[\varsigma_0\right](\widehat{\varsigma}_0)
 \right\} \\
 \Gamma[C(x, \, K, \, T,\, \varsigma_0)](\widehat{\varsigma}_0) & = & x^2 \; \frac{e^{- d_1^2}}{2 \pi} \,  T \,
 \Gamma\left[\varsigma_0\right](\widehat{\varsigma}_0)
 \end{eqnarray}
where 
$$
d_1 = \frac{\ln x - \ln K + \frac{\widehat{\varsigma}_0^2}{2} T
}{\widehat{\varsigma}_0 \,  \sqrt{T} }, \; \; d_2 = d_1 - \widehat{\varsigma}_0\,  \sqrt{T},
$$
and $C(x,\, K, \, T,\, \varsigma_0)$ 
denotes the price of the call with strike $K$, maturity $T$ when the
 underlying has value $x$ and the trader known the volatility estimator  $\varsigma_0$.
Moreover we have the following bid and ask prices:
\begin{eqnarray}
C_{\text{ask}}(x, \, K, \, T,\, \widehat{\varsigma}_0) &= & x\, \mathcal{G} (d_1) - K \, \mathcal{G}(d_2)   
 +x \, \sqrt{ \,T} \; \frac{e^{-\frac{1}{2} d_1^2}}{\sqrt{2
\pi}} \, \sqrt{  \Gamma\left[\varsigma_0 \right] (\widehat{\varsigma}_0) } \mathcal{N}_{1-\alpha}
\label{BS:ask-price} \\
& & + x \,\sqrt{T} \;  
\frac{e^{-\frac{1}{2} d_1^2}}{\sqrt{2
\pi}}  \; \left\{  \mathcal{A} \left[\varsigma_0 \right] (\widehat{\varsigma}_0) +
 \frac{d_1 \, d_2}{2 \, \sigma_0} \,  \Gamma\left[\varsigma_0\right](\widehat{\varsigma}_0)
  \right\} \nonumber  \\ 
 & & \nonumber \\ 
C_{\text{bid}}(x, \, K, \, T,\, \widehat{\varsigma}_0) & = & x\, \mathcal{G} (d_1) - K \, \mathcal{G}(d_2) 
+ x\,  \sqrt{ T}  \; \frac{e^{-\frac{1}{2} d_1^2}}{\sqrt{2
\pi}} \, \sqrt{ \Gamma\left[\varsigma_0\right](\widehat{\varsigma}_0)
 } \mathcal{N}_{\alpha} \label{BS:bid-price} \\ 
& &  + x\,\sqrt{T} \; 
\frac{e^{-\frac{1}{2} d_1^2}}{\sqrt{2
\pi}} \; \left\{  \mathcal{A} \left[\varsigma_0 \right](\widehat{\varsigma}_0) +
 \frac{d_1 \, d_2}{2 \,\widehat{\varsigma}_0} \, \Gamma\left[\varsigma_0\right](\widehat{\varsigma}_0)
 \right\},  \nonumber 
\end{eqnarray}
where $\mathcal{G}$ denotes the cumulative distribution function of the reduced Gaussian law
and $\mathcal{N}_{\alpha}$ is defined in Proposition \ref{PROP:OP}.

\end{corollario}

\begin{proof}
We know that the cost of hedging strategy $F(\widehat{\varsigma}_0)$ and the portfolio 
$\Delta(\widehat{\varsigma}_0)$ of a call on Black and Scholes model are given by
\begin{eqnarray*}
 C(x, \, K, \, T,\, \widehat{\varsigma}_0)  &= & C(\widehat{\varsigma}_0) = x\, \mathcal{G} (d_1) - K \,
 \mathcal{G}(d_2) \\
 \Delta(x, \, K, \, T,\, \widehat{\varsigma}_0)  &= & 
 \frac{\partial C}{\partial x}(\widehat{\varsigma}_0) = \mathcal{G} (d_1).
\end{eqnarray*}
We can easily calculate the two first derivatives of $F(\widehat{\varsigma}_0, \, x)$ and 
$\Delta(\widehat{\varsigma}_0, \, x)$ with respect to the volatility $\widehat{\varsigma}_0$, then
the two equations (\ref{BS:ask-price}) and (\ref{BS:bid-price}) come directly from a 
computation of $\mathcal{A}[C(x, \, K, \, T,\,\varsigma_0)](\widehat{\varsigma}_0)$
and $\Gamma[C(x, \, K, \, T,\,\varsigma_0)](\widehat{\varsigma}_0)$ using Corollary \ref{theorem-corollary}.
\end{proof}

\subsection{Analysis of the impact of parameter uncertainty}

We now analyze the correction on the pricing formula due to the presence of an uncertainty 
on the volatility parameter. We have the following corollaries that can be proved with simple 
computations.
We analyze only the mid-price, i.e. the average price of the bid and ask prices. We also use 
the notation
$$
r_r^{BS}\left( \widehat{\varsigma}_0  \right) = 2 \, \widehat{\varsigma}_0  \, \frac{ 
\mathcal{A} \left[\varsigma_0 \right](\widehat{\varsigma}_0)}{
 \Gamma\left[\varsigma_0 \right] (\widehat{\varsigma}_0) }\, .
$$

\begin{corollario}[Delta and Gamma correction]\label{corollary-delta-gamma}\hfill
\vspace{0.2cm}

The bias on the call price due to the uncertainty on volatility verifies 
\begin{equation}
 \frac{\partial \mathcal{A}[C]}{\partial K}  
= \frac{d_1 \; \mathcal{A}[C]}{K \, \widehat{\varsigma}_0 \,  \sqrt{T}} -
\frac{x}{2 \,  K \,  \widehat{\varsigma}_0^2} \frac{e^{-\frac{1}{2} d_1^2}}{\sqrt{2
\pi}} \,  (d_1 + d_2) \,  \Gamma\left[\widehat{\varsigma}_0\right].
\end{equation}
Moreover, the bias and its first derivative with respect to K are positive at the money 
if  $r_r^{BS}\left(\widehat{\varsigma}_0\right)>\frac{1}{4}\widehat{\varsigma}_0^2\,T$ and
are negative when the stricktly inequality is reversed. 
Finally, the bias is convex in K at the money if and only if
$$
r_r^{BS}\left( \widehat{\varsigma}_0 \right) <   
\frac{\widehat{\varsigma}_0^4 \, T^2 + 4 \widehat{\varsigma}_0^2 \, T + 32}{4 \widehat{\varsigma}_0^2 \, T + 16}\; .
$$
\end{corollario}

\begin{corollario}[Time evolution of the correction at the money]\hfill
\vspace{0.2cm}

The third cross derivative of the bias, twice with respect to the strike K and once
 with respect to the maturity, i.e. $\displaystyle
 \frac{\partial^3 \mathcal{A}[C] }{\partial K^2 \partial T}$,
 is positive if and only if 
$$
    r_r^{BS}\left(\widehat{\varsigma}_0  \right)  > \frac{1}{4} \; \frac{ \widehat{\varsigma}_0^2 \,  T
      \left(\widehat{\varsigma}_0^2 \,  T -4\right)^2 +128 }{ 16 + \widehat{\varsigma}_0^4 \, T^2 } \; .
$$
\end{corollario}

The two previous corollaries have an interesting consequence. 

\begin{proposizione}[Smile on implied volatility]\hfill
\vspace{0.2cm}

There exists an interval $]a, \, b [$ such that for all $ r_r^{BS}(\widehat{\varsigma}_0) \in\, 
 ]a, \, b [$, the implied BS volatility recognized from the average price $F_{\text{mid}}$ 
 given by Corollary \ref{pricing-BS} is a convex function around the money.  
 Furthermore, the second derivative of the implied volatility with respect to the strike is a
 decreasing function of the maturity $T$.  That is, the implied volatility exhibit the smile
 effect and this effect is stronger for short maturities than long one.
\end{proposizione}

\begin{proof}
We start fixing by $r_r^{BS}(\widehat{\varsigma}_0) = \frac{1}{4}\widehat{\varsigma}_0^2 \, T$. 
Thanks to Corollary 
\ref{corollary-delta-gamma}, we have that $\mathcal{A}[C]= \frac{\partial 
\mathcal{A}[C]}{\partial K}=0$ at the money. In addiction, the same corollary assures us 
that the bias $\mathcal{A}[C]$ is strictly convex around the money thanks the 
continuity of the derivatives.  For the same value of $r_r^{BS}(\widehat{\varsigma}_0 )$
 we easily check that the cross derivative of the bias two times w.r.t. $K$ 
 and one time w.r.t. $T$ is negative at the money. 
We argue that the call price at the money is equal to the Black-Scholes one 
 whereas around the money the call price in our model is bigger than the Black-Scholes one. 
Giving that the Black-Scholes formula is used to find the implied volatility, 
we conclude that the implied 
volatility is strictly convex around the money and that this convexity is more marked for short 
maturities. We now remark that the implied volatility belongs to $C^{\infty}$    
as function of the strike $K$ since it is a composition of infinite differentiable function. 
We conclude that, for each sufficiently small neighborhood of the money,  it exists an interval 
$]a, \, b [$,  including  $\frac{1}{4}\widehat{\varsigma}_0^2 \, T$, such that the implied volatility remains 
strictly convex and with convexity decreasing with the maturity. 
\end{proof}

\section{Example: Constant Elasticity Volatility Model}

In this section, we present a second example. We consider the CEV model, see Cox and Ross \cite{bib:Cox}.
The volatility function is equal to $\sigma S_t^{\beta-1}$ where $\sigma$ is a positive constant and $\beta$ is constant 
and belongs to $(0, \,1)$. As in the previous example, we suppose for sake of simplicity that the risky asset is a martingale under 
historical probability $\mathbb{P}_1$. 
We assume then that $S_t$ is bounded from below by a positive constant. 
The diffusion is then
$$
dS_t =\sigma \, S_t^{\beta} dW_t
$$
This diffusion admits a closed form solution via squared Bessel process theory, see for instance Jeanblanc et al. 
\cite{bib:Jeanblanc}. The previous SDE can be rewritten in the following form:
$$
dS_t = e^{\log \sigma +\beta \log S_t} dW_t
$$
In this form, it is easy to adapt our previous result to CEV diffusion.
In our case, we consider that the couple of parameters $(\sigma, \beta)$ is fixed but unknown. The option seller knows 
only the couple of estimators $(\varsigma, B)$ and the estimated values $(\widehat{\varsigma}, \widehat{B})$. 
We fix an error structure for the couple of the parameters:
$$
(\mathbb{R}^+\times(0,1), \mathcal{B}(\mathbb{R}^+ \times (0,1)), \widetilde{\mathbb{P}}, \mathbf{d}, \Gamma)
$$ 
The explicit form of the probability measure $ \widetilde{\mathbb{P}}$ and of the carr\'{e} du champ operator are not compulsory 
since we need only to compute the bias and the variance at the point  $(\widehat{\varsigma}, \widehat{B})$.
We now show how to compute in our framework the option price in a particular case, that is  a regular but non-linear payoff. 

\subsection{Variance and bias of the diffusion process}

We consider the diffusion process from the point of view of the option seller:
\begin{equation}\label{SDE:CEV_X}
dX^{(\widehat{\varsigma}, \widehat{B})}_t = \widehat{\varsigma}\, 
\left[ X_t^{(\widehat{\varsigma}, \widehat{B})} \right]^{\widehat{B}} dW_t
\end{equation}

We have the first two results  about the sharp and the variance of $X^{(\varsigma, \widehat{B})}$:

\begin{proposizione}[Sharp of the process $X^{(\varsigma, \widehat{B})}$]\label{lemma-sharp-cev}\hfill
\vspace{0.2cm}

The sharp of $X^{(\varsigma, \widehat{B})}$ verifies the following SDE
\begin{equation}\label{SDE-sharp-CEV}
d\left(X^{(\varsigma, B)}\right)^{\#}_t = \varsigma B \left(X^{(\varsigma, B)}\right)_t^{\#} \left[X^{(\varsigma, B)}_t\right]^{B-1} dW_t
+ \left[X^{(\varsigma, B)}_t \right]^{B} \left( \varsigma^{\#} + B^{\#} \varsigma \right) dW_t   \, .
\end{equation}
Moreover, the previous solution admits the following closed form:
\begin{equation}\label{cosed-form-sharp-CEV}
\left(X^{(\varsigma, B)}\right)^{\#}_t =  \left(\varsigma^{\#} + B^{\#} \varsigma  \right) M_t 
\left[ \int_0^t   \left[X^{(\varsigma, B)}_s\right]^{A} M_s^{-1} dW_s - \varsigma B \int_0^t   \left[X^{(\varsigma, B)}_s
\right]^{2B-1} M_s^{-1} ds \right]  \, ,
\end{equation}
where $\displaystyle M_t = \mathcal{E}\left(\int_0^{\, (\cdot)}  \varsigma B  \left[X^{(\varsigma, B)}_s\right]^{B-1} dW_s\right)_t$.
\end{proposizione}

\begin{proof} We omit the dependency of $X^{(\varsigma, B)}$ on the couple $(\varsigma, B)$  for simplicity's sake.
We start discretizing the SDE (\ref{SDE:CEV_X}) using Euler's scheme. Let $\Pi_n = \{0=t^n_0<t^n_1<\ldots <t_n^n =T\}$ be
a subdivision on the interval $[0,T]$ for each integer $n$ with the discretization step denoted by $\delta_n = \sup_i(t_{i+1}-t_i)$. 
We have 
$$
\left\{
\begin{array}{rcl}
\displaystyle X_{t}^n &  = & \displaystyle  X_0 \quad \quad \quad\quad\quad\quad  \quad
\quad\quad\quad\quad  \quad\quad\quad\quad \;  \, \forall t\in[t_0^n, t_1^n) \\
\displaystyle X_{t_i^n}^n & =  & \displaystyle X^n_{t_{i-1}^n} + \varsigma \left(  X^n_{t_{i-1}^n} \right)^{B}
\left( W_{t_i^n} -W_{t_{i-1}^n} \right)   \quad \;  \, \forall i=1,\ldots n   \\
\displaystyle X_{t}^n &  = & \displaystyle  X^n_{t_i^n} \quad \quad \quad\quad\quad\quad  \quad
\quad\quad\quad\quad  \quad\quad\quad \quad \forall t\in[t_i^n, t_{i+1}^n)  \, .
\end{array}
\right.
$$
We now apply the sharp operator to our discretized process $X^n$. On the first interval $[t_0^n, t_1^n)$, the process $X^n$ 
does not depend on the couple $(\varsigma, B)$ then $X_t^{\#} =0$ for any $t\in [t_0^n, t_1^n)$.
At time $t_i$, we have 
$$
\left(X_{t_i^n}^n\right)^{\#}  =   \left(X^n_{t_{i-1}^n}\right)^{\#} + \left(\varsigma \left(  X^n_{t_{i-1}^n} \right)^{B}
\left( W_{t_i^n} -W_{t_{i-1}^n} \right) \right)^{\#}
$$
thanks to the linearity of the sharp. The second term of the previous relation is the sharp of a function depending both on the
couple $(\varsigma, B)$ and on the process $X^n$. In order to compute it we apply the second property of sharp operator, see 
proposition \ref{prop:sharp} and proposition V.8 page 83 in 
\cite{bib:Bouleau-erreur} for a more detailed analysis. We have for any regular function f
$$
 \left(f\left(\varsigma, B,  X^n_{t_{i-1}^n} \right) \right)^{\#} = \frac{\partial f}{\partial x_1}
 \left(\varsigma, B,  X^n_{t_{i-1}^n} \right) \varsigma^{\#} + \frac{\partial f}{\partial x_2}
 \left(\varsigma, B,  X^n_{t_{i-1}^n} \right) B^{\#} +  \frac{\partial f}{\partial x_3}
 \left(\varsigma, B,  X^n_{t_{i-1}^n} \right) \left(X^n_{t_{i-1}^n} \right)^{\#} \, .
$$
where $\frac{\partial }{\partial x_j}$ denotes the partial derivatives w.r.t. the jth variable. 
In our case we find
\begin{eqnarray*}
 \left(\varsigma \left(  X^n_{t_{i-1}^n} \right)^{B}
\left( W_{t_i^n} -W_{t_{i-1}^n} \right) \right)^{\#} & = &   \left(  X^n_{t_{i-1}^n} \right)^{B}
\left( W_{t_i^n} -W_{t_{i-1}^n} \right) \varsigma^{\#}  \\
& & + \varsigma  \left(  X^n_{t_{i-1}^n} \right)^{B}
\left( W_{t_i^n} -W_{t_{i-1}^n} \right) B^{\#}  \\
& & + \varsigma B  \left(  X^n_{t_{i-1}^n} \right)^{B-1}
\left( W_{t_i^n} -W_{t_{i-1}^n} \right) \left(  X^n_{t_{i-1}^n} \right)^{\#}
\end{eqnarray*}
Then, we have
\begin{eqnarray*}
\left(X_{t_i^n}^n\right)^{\#} & = & \left(X^n_{t_{i-1}^n}\right)^{\#}+   \left(  X^n_{t_{i-1}^n} \right)^{B}
\left( W_{t_i^n} -W_{t_{i-1}^n} \right) \varsigma^{\#}  \\
& & + \varsigma  \left(  X^n_{t_{i-1}^n} \right)^{B}
\left( W_{t_i^n} -W_{t_{i-1}^n} \right) B^{\#}  \\
& & + \varsigma B  \left(  X^n_{t_{i-1}^n} \right)^{B-1}
\left( W_{t_i^n} -W_{t_{i-1}^n} \right) \left(  X^n_{t_{i-1}^n} \right)^{\#}
\end{eqnarray*}
and $\left(X_{t}^n\right)^{\#} = \left(X_{t_i^n}^n\right)^{\#}$ for any $t\in[t_i^n,t_{i+1}^n)$. 
We remark that $\left(X^n\right)_t^{\#}$ is a Picard sequence that verifies usual conditions.
As a consequence it exists an unique process $\left(X\right)_t^{\#}$ defined into the probability space $\mathbb{P}\times 
\overline{\mathbb{P}}$ that is the limit of the sequence $\left(X_{t}^n\right)^{\#}$ when $\delta_n$ goes to $0$. 
This limit is the sharp of the process $X_t$
thanks to the closedness of the sharp operator $(\cdot)^{\#}$.
Moreover,  $\left(X\right)_t^{\#}$ verifies
SDE (\ref{SDE-sharp-CEV})
Finally, (\ref{SDE-sharp-CEV}) is an affine SDE, then we can apply the method of constant variation, see for instance 
\cite{bib:Protter}, to find a closed form solution. It is easy to check that equation (\ref{cosed-form-sharp-CEV}) is the
solution of SDE  (\ref{SDE-sharp-CEV}).
\end{proof}

\begin{corollario}[Variance of the process $X^{(\varsigma, B)}$]\label{corr-variance-cev}\hfill
\vspace{0.2cm}

The carr\'{e} du champ operator acts on $X^{(\varsigma, B)}$ as
\begin{equation}\label{cosed-form-variance-CEV}
\begin{array}{rcl}
\displaystyle \Gamma\left[X^{(\varsigma, B)}_t\right](\widehat{\varsigma}, \widehat{B}) & = & \displaystyle  M_t^2 
\left[ \int_0^t   \left[X^{(\widehat{\varsigma}, \widehat{B})}_s\right]^{\widehat{B}} M_s^{-1} dW_s - \widehat{\varsigma} 
\widehat{B} \int_0^t   \left[X^{(\widehat{\varsigma}, \widehat{B})}_s
\right]^{2\widehat{B}-1} M_s^{-1} ds \right]^2  \\
& & \\
& &\displaystyle \times \left\{ \Gamma \left[ \varsigma \right](\widehat{\varsigma}) +
2 \widehat{\varsigma} \,   \Gamma \left[ \varsigma,B  \right](\widehat{\varsigma},\widehat{B}) 
+ \widehat{\varsigma}^2 \,   \Gamma[B](\widehat{B})  \right\}
\end{array}
\end{equation}

\end{corollario}

\begin{proof}
Thank to the first property of sharp operator, see 
proposition \ref{prop:sharp} we have that
$$
\Gamma\left[X^{(\varsigma, B)}_t\right] = \overline{\mathbb{E}} \left[ \left\{\left(X^{(\varsigma, B)}\right)^{\#}_t \right\}^2  \right]  
$$
Using the closed form (\ref{cosed-form-sharp-CEV}) and recalling that only $\varsigma^{\#}$ and $B^{\#}$ depend on 
$\overline{\mathbb{P}}$, we have that 
\begin{equation*}
\begin{array}{rcl}
\displaystyle \Gamma\left[X^{(\varsigma, B)}_t\right] & = & \displaystyle  M_t^2 
\left[ \int_0^t   \left[X^{(\varsigma, B)}_s\right]^{B} M_s^{-1} dW_s - \varsigma B \int_0^t   \left[X^{(\varsigma, B)}_s
\right]^{2B-1} M_s^{-1} ds \right]^2  \\
& & \\
& &\displaystyle \times  \overline{\mathbb{E}} \left[ \left( \varsigma^{\#} + \varsigma B^{\#} \right)^2 \right]
\end{array}
\end{equation*}
The last expectation is easy since $\overline{\mathbb{E}} \left[ \left( \varsigma^{\#} \right)^2 \right] = \Gamma[\varsigma]$ and
$\overline{\mathbb{E}} \left[ \left( B^{\#} \right)^2 \right] = \Gamma[B]$ by definition of sharp.
Finally, thanks to the continuity of the carr\'{e} du champ operator $\Gamma$ at points $\widehat{\varsigma}$ and
$\widehat{B}$ and using  lemma \ref{Lemma-Gateaux-2}, we find equation (\ref{cosed-form-variance-CEV}).

\end{proof}

\begin{corollario}[Covariance of the process $X^{(\varsigma, B)}$]\label{corr-variance-cev}\hfill
\vspace{0.2cm}

We have the two following covariances
\begin{equation}\label{cosed-form-covariance-1-CEV}
\begin{array}{rcl}
\displaystyle \Gamma\left[\varsigma, X^{(\varsigma, B)}_t\right](\widehat{\varsigma}, \widehat{B}) & = & \displaystyle  M_t 
\left[ \int_0^t   \left[X^{(\widehat{\varsigma}, \widehat{B})}_s\right]^{\widehat{B}} M_s^{-1} dW_s - \widehat{\varsigma} 
\widehat{B} \int_0^t   \left[X^{(\widehat{\varsigma}, \widehat{B})}_s
\right]^{2\widehat{B}-1} M_s^{-1} ds \right]  \\
& & \\
& &\displaystyle \times \left\{ \Gamma \left[ \varsigma \right](\widehat{\varsigma}) +
 \widehat{\varsigma} \,   \Gamma \left[ \varsigma, B \right](\widehat{\varsigma},\widehat{B}) \right\} 
\end{array}
\end{equation}

\begin{equation}\label{cosed-form-covariance-2-CEV}
\begin{array}{rcl}
\displaystyle \Gamma\left[B, X^{(\varsigma, B)}_t\right](\widehat{\varsigma}, \widehat{B}) & = & \displaystyle  M_t 
\left[ \int_0^t   \left[X^{(\widehat{\varsigma}, \widehat{B})}_s\right]^{\widehat{B}} M_s^{-1} dW_s - \widehat{\varsigma} 
\widehat{B} \int_0^t   \left[X^{(\widehat{\varsigma}, \widehat{B})}_s
\right]^{2\widehat{B}-1} M_s^{-1} ds \right]  \\
& & \\
& &\displaystyle \times  \left\{ \widehat{\varsigma} \,  \Gamma \left[ B \right](\widehat{B}) +
  \Gamma \left[ \varsigma, B \right](\widehat{\varsigma},\widehat{B}) \right\} 
\end{array}
\end{equation}
\end{corollario}

The proof is similar to the previous one.

\begin{proposizione}[Bias of the process $X^{(\varsigma, B)}$]\label{corr-bias-cev}\hfill
\vspace{0.2cm}

The bias of $X^{(\varsigma, A)}$ verifies the following SDE
\begin{equation}\label{SDE-bias-CEV}
\begin{array}{rcl}
\displaystyle d\mathcal{A}\left[X^{(\varsigma, B)}\right]_t & = &\displaystyle \varsigma B 
\left[X^{(\varsigma, B)}_t\right]^{B-1} \,  \mathcal{A}\left[X^{(\varsigma, B)}\right]_t  dW_t
+ \left[X^{(\varsigma, B)}_t \right]^{A} \left( \mathcal{A}[\varsigma] + \varsigma \,  \mathcal{A}[B] \right) dW_t   \\
& & \\
& & \displaystyle+ \frac{1}{2}   \varsigma B(B-1) \left[X^{(\varsigma, B)}_t\right]^{B-2} 
\Gamma \left[X^{(\varsigma, A)}\right]_t dW_t +  \left[X^{(\varsigma, B)}_t \right]^{B} \, \Gamma[\varsigma, B]  dW_t  \\
& & \\
& & \displaystyle+ \left[X^{(\varsigma, B)}_t\right]^{B-1}  \left( B\, \Gamma\left[\varsigma, X^{(\varsigma, B)}\right]_t 
+ \varsigma \, \Gamma \left[B, X^{(\varsigma, B)}\right]_t  \right) dW_t \\
& & \\
& & \displaystyle+ \frac{1}{2}  \varsigma \left[X^{(\varsigma, B)}_t \right]^{B}  \Gamma[B] dW_t
\end{array}
\end{equation}
Moreover, the previous solution admits the following closed form:
\begin{equation}\label{cosed-form-bias-CEV}
\mathcal{A}\left[X^{(\varsigma, B)}_t\right](\widehat{\varsigma},\widehat{B}) =  M_t 
\left[ \int_0^t   \nu(s, X_s, \widehat{\varsigma}, \widehat{B}) M_s^{-1} dW_s +  \int_0^t  d
\left[  \nu(\cdot, X, \widehat{\varsigma},\widehat{B}), M^{-1} \right]_s \right]  \, ,
\end{equation}
where 
\begin{eqnarray*}
\nu(s, X_s, \widehat{\varsigma}, \widehat{B}) &=& \left[X^{(\widehat{\varsigma}, \widehat{B})}_s \right]^{\widehat{B}} 
\left( \,  \mathcal{A}[\varsigma](\widehat{\varsigma}) + \widehat{\varsigma} \,  \mathcal{A}[B](\widehat{B}) \right)
+ \frac{1}{2}  \widehat{\varsigma} \left[X^{(\widehat{\varsigma}, \widehat{B})}_s \right]^{\widehat{B}}  
\Gamma[B](\widehat{B}) \\
& & + \frac{1}{2}   \widehat{\varsigma} \widehat{B}(\widehat{B}-1) 
\left[X^{(\widehat{\varsigma}, \widehat{B})}_s\right]^{\widehat{B}-2}  \Gamma \left[X^{\varsigma, B)}_s\right]
(\widehat{\varsigma}, \widehat{B}) +  \left[X^{(\widehat{\varsigma}, \widehat{B})}_s \right]^{\widehat{B}} \, 
\Gamma[\varsigma, B](\widehat{\varsigma}, \widehat{B})  \\
& & \displaystyle + \left[X^{(\widehat{\varsigma}, \widehat{B})}_s\right]^{\widehat{B}-1}  \left\{ \widehat{B}\, 
\Gamma\left[\varsigma, X^{(\varsigma, B)}_s\right](\widehat{\varsigma}, \widehat{B}) 
+ \widehat{\varsigma} \, \Gamma \left[B, X^{(\varsigma, A)}_s\right](\widehat{\varsigma},\widehat{B})   \right\}
\end{eqnarray*}
\end{proposizione}

\begin{proof}
The proof is based on the same arguments used to prove proposition \ref{lemma-sharp-cev}
We discretize the SDE (\ref{SDE:CEV_X}) using the same Euler's scheme. 
Then, we apply the bias operator $\mathcal{A}$ at the discretized process $X^n_{t}$.
The chain rule (\ref{bias-chain-rule}) verified by the generator depends both on the bias and the variance.
Using the chain rule and the linearity we have
\begin{eqnarray*}
\mathcal{A}\left[X_{t_i^n}^n\right] & = & \mathcal{A}\left[X^n_{t_{i-1}^n}\right]+   \left(  X^n_{t_{i-1}^n} \right)^{B}
\left( W_{t_i^n} -W_{t_{i-1}^n} \right) \mathcal{A}[\varsigma]   + \varsigma  \left(  X^n_{t_{i-1}^n} \right)^{B}
\left( W_{t_i^n} -W_{t_{i-1}^n} \right) \mathcal{A}[B]  \\
& & + \varsigma B  \left(  X^n_{t_{i-1}^n} \right)^{B-1}
\left( W_{t_i^n} -W_{t_{i-1}^n} \right) \mathcal{A}\left[  X^n_{t_{i-1}^n} \right] +  
\frac{1}{2}  \varsigma B  \left(  X^n_{t_{i-1}^n} \right)^{B} \left( W_{t_i^n} -W_{t_{i-1}^n} \right) \Gamma \left[  B \right] \\
& & + \frac{1}{2} 
\varsigma B (B-1) \left(  X^n_{t_{i-1}^n} \right)^{B-2}
\left( W_{t_i^n} -W_{t_{i-1}^n} \right) \Gamma \left[  X^n_{t_{i-1}^n} \right]  \\
& & +  \varsigma B  \left(  X^n_{t_{i-1}^n} \right)^{B-1} \left( W_{t_i^n} -W_{t_{i-1}^n} \right) \Gamma\left[B, X^n_{t_{i-1}^n}\right] \\
& & + B  \left(  X^n_{t_{i-1}^n} \right)^{B-1} \left( W_{t_i^n} -W_{t_{i-1}^n} \right) \Gamma\left[\varsigma, X^n_{t_{i-1}^n}\right]
+  \left(  X^n_{t_{i-1}^n} \right)^{B} \left( W_{t_i^n} -W_{t_{i-1}^n} \right)  \Gamma[B, \varsigma]
\end{eqnarray*} 
We recall that $\Gamma$ is a closed operator, then $\Gamma \left[  X^n_t \right]$ (resp. 
$ \Gamma\left[B, X^n_t\right]$, $ \Gamma\left[\varsigma, X^n_t\right]$) converges to 
$\Gamma \left[  X_t \right]$ (resp.  $ \Gamma\left[B, X_t\right]$, $ \Gamma\left[\varsigma, X_t\right]$).
Then, we have construct a Picard sequence $\mathcal{A}\left[X_t^n\right]$ that verifies usual conditions.
As a consequence,  it exists an unique process $\mathcal{A}\left[X\right]_t$  defined into the probability space $\mathbb{P}$ 
that is the limit of the sequence $\mathcal{A}\left[X_{t}^n\right]$ when $\delta_n$ goes to $0$. 
This limit is the bias of the process $X_t$
thanks to the closedness of the bias operator $\mathcal{A}$.
Moreover,  $\mathcal{A}\left[X\right]_t$  verifies SDE (\ref{SDE-bias-CEV}) that is affine w.r.t.  $\mathcal{A}\left[X\right]_t$. 
Applying Ito formula it is easy to check that equation (\ref{cosed-form-bias-CEV}) is the solution of SDE (\ref{SDE-bias-CEV}),
where we have used the continuity of $\mathcal{A}$ and $\Gamma$ at point $(\widehat{\varsigma}, \widehat{B})$ and lemma
\ref{Lemma-Gateaux-2}.
\end{proof}

Finally, we can resume some interesting results about the variance and the bias of the process $X_t^{(\widehat{\varsigma}, \widehat{B})}$, that can be deduced easily from the previous results.

\begin{proposizione}[Properties of the variance and the bias of 
$X_t^{(\widehat{\varsigma}, \widehat{B})}$]\label{[properties-cev}\hfill
\vspace{0.2cm}

We have the following properties:
\begin{enumerate}
\item $\displaystyle \Gamma\left[X^{(\varsigma, A)}_t\right](\widehat{\varsigma}, \widehat{B})$ is a non-negative 
local-sub-martingale taking value $0$ at time $0$.
\item $\Gamma\left[\varsigma, X^{(\varsigma, A)}_t\right](\widehat{\varsigma}, \widehat{B})$ and 
$\Gamma\left[A, X^{(\varsigma, A)}_t\right](\widehat{\varsigma}, \widehat{B})$ are local-martingale 
taking value $0$ at time $0$.
\item $\mathcal{A}\left[X^{(\varsigma, A)}_t\right](\widehat{\varsigma},\widehat{B})$ is a local martingale.
\end{enumerate}
\end{proposizione}

\subsection{Financial aspects}

The easiest example is an option that pays $S_T^2$ at maturity.
The price without uncertainty can be easily computed using the exact law of $S_T$ by means of squared Bessel process 
approach or applying  Monte-Carlo techniques thanks to Euler discretization.  
In both cases, we have the (approximated) law of the process $S$ for any time $t\in[0,T]$.
Suppose now that the option seller does known the true value of the couple $(\sigma, \beta)$ but only the estimated value
$(\widehat{\varsigma}, \widehat{B})$ and the estimators $(\varsigma, B)$. 
The option seller has to work with his estimated process $X_t^{(\widehat{\varsigma}, \widehat{B})}$.
Clearly the hedging cost can be estimated using the exact law given by Bessel process theory or thanks to Euler discretization.
In both cases, it is possible to estimate the bias and the variance-covariance of $X_t^{(\varsigma, B)}$ depending on the
uncertainty on $(\varsigma, B)$.

\underline{\emph{Exact law.}} Given the exact law of $X_t^{(\widehat{\varsigma}, \widehat{B})}$, we can compute the
joint law of $\Gamma\left[X^{(\varsigma, B)}_t\right](\widehat{\varsigma}, \widehat{B})$,  
$\Gamma\left[\varsigma, X^{(\varsigma, B)}_t\right](\widehat{\varsigma}, \widehat{B})$, 
$\Gamma\left[B, X^{(\varsigma, B)}_t\right](\widehat{\varsigma}, \widehat{B})$ and
$\mathcal{A}\left[X^{(\varsigma, B)}_t\right](\widehat{\varsigma},\widehat{B})$
using equations (\ref{cosed-form-variance-CEV}), (\ref{cosed-form-covariance-1-CEV}),  
(\ref{cosed-form-covariance-2-CEV}) and 
(\ref{cosed-form-bias-CEV}).

\underline{\emph{Euler scheme.}} We recall that the operators carr\'{e} du champ $\Gamma$ and generator $\mathcal{A}$ are
closed. Then we can define an Euler scheme to find the approximated joint law of 
$\Gamma\left[X^{(\varsigma,B)}_t\right](\widehat{\varsigma}, \widehat{B})$,  
$\Gamma\left[\varsigma, X^{(\varsigma, B)}_t\right](\widehat{\varsigma}, \widehat{B})$, 
$\Gamma\left[B, X^{(\varsigma, B)}_t\right](\widehat{\varsigma}, \widehat{B})$ and
$\mathcal{A}\left[X^{(\varsigma, B)}_t\right](\widehat{\varsigma},\widehat{B})$.
Moreover, we have used the same approach to find the SDE verified by the sharp and the bias. Then, we can use a unique
partition $\Pi_n$ of the interval $[0,T]$ in order to estimate the hedging cost and the four processes  
$\Gamma\left[X^{(\varsigma, B)}_t\right](\widehat{\varsigma}, \widehat{B})$,  
$\Gamma\left[\varsigma, X^{(\varsigma, B)}_t\right](\widehat{\varsigma}, \widehat{B})$, 
$\Gamma\left[A, X^{(\varsigma, B)}_t\right](\widehat{\varsigma}, \widehat{B})$ and
$\mathcal{A}\left[X^{(\varsigma, B)}_t\right](\widehat{\varsigma},\widehat{B})$.

We now analyze the bias and the variance of the square of $X_t^{(\varsigma, B)}$.

\begin{proposizione}[Variance and bias of $\left(X_t^{(\varsigma, B)}\right)^2$]\label{[square-cev}\hfill
\vspace{0.2cm}

We have the two following relations
\begin{eqnarray}
\Gamma\left[\left(X^{(\varsigma, B)}_t \right)^2\right](\widehat{\varsigma}, \widehat{B}) & = & 
4 \left(X_t^{(\widehat{\varsigma}, \widehat{B})} \right)^2 \, \Gamma\left[X^{(\varsigma, B)}_t\right]
(\widehat{\varsigma}, \widehat{B}) \\
\mathcal{A}\left[\left(X^{(\varsigma, B)}_t \right)^2\right](\widehat{\varsigma}, \widehat{B}) & = &
2 X_t^{(\widehat{\varsigma}, \widehat{B})} \mathcal{A}\left[X^{(\varsigma, B)}_t\right](\widehat{\varsigma},\widehat{B})
+ \Gamma\left[X^{(\varsigma, B)}_t\right] (\widehat{\varsigma}, \widehat{B})
\end{eqnarray}
\end{proposizione}

The proof is a direct application of the two chain rules (\ref{functional-calculus}) and (\ref{bias-chain-rule})
verified by operators $\Gamma$ and $\mathcal{A}$.

We can now conclude with the price of the power option in the case of uncertainty.

\begin{proposizione}[Bid and ask prices of the power option]\label{[prices-cev}\hfill
\vspace{0.2cm}

We have the following bid and ask prices for the power option in accord with Principle \ref{principle-1}:
\begin{eqnarray}
P_{\text{ask}}(x,T, \widehat{\sigma}, \widehat{B}) &=& 
\mathbb{E}_1\left[\left(X^{(\widehat{\varsigma}, \widehat{B})}_T\right)^2\right]
+ \mathbb{E}_1\left[ \mathcal{A}\left[\left(X^{(\varsigma, B)}_t \right)^2\right](\widehat{\varsigma}, \widehat{B})  \right] \\
& & + \sqrt{\mathbb{E}_1\left[ \Gamma\left[\left(X^{(\varsigma, B)}_t \right)^2\right](\widehat{\varsigma},
 \widehat{B})\right] } \, \mathcal{N}_{1-\alpha} \nonumber \\
 & & \nonumber \\
  & & \nonumber \\
P_{\text{bid}}(x,T, \widehat{\sigma}, \widehat{B}) &=& 
\mathbb{E}_1\left[\left(X^{(\widehat{\varsigma}, \widehat{B})}_T\right)^2\right]
+ \mathbb{E}_1\left[ \mathcal{A}\left[\left(X^{(\varsigma, B)}_t \right)^2\right](\widehat{\varsigma}, \widehat{B})  \right] \\
& & + \sqrt{\mathbb{E}_1\left[ \Gamma\left[\left(X^{(\varsigma, B)}_t \right)^2\right](\widehat{\varsigma}, 
\widehat{B}) \right]}  \, \mathcal{N_{\alpha}}  \nonumber
\end{eqnarray}

where $\mathcal{N}_{\alpha}$ are defined in Proposition \ref{PROP:OP}.
\end{proposizione}

The proposition is a direct consequence of Proposition \ref{PROP:OP}.

\section{Acknowledgements}

I am grateful to Nicolas Bouleau for introducing me to the subject of error theory using Dirichlet forms and for many discussions. I sincerely thank Robert Dalang for reading the previous versions
of this article, for proving many helpful comments and for useful discussions. 
I would like to thank two referees for their valuable comments.


\begin{thebibliography}{99}

\bibitem{bib:Albeverio} Albeverio, S. (2003): {\it Theory of
    Dirichlet Forms and Applications}, Springer-Verlag, Berlin.

\bibitem{bib:Avelaneda} Avellaneda, M.; Levy, A. and Paras, A. (1995): {\it  Pricing and Hedging 
Derivative Securities in Markets with Uncertain Volatilities} Applied Mathematical Finance, 2,
 73-88. 



\bibitem{bib:Bellamy} Bellamy, N. and Jeanblanc, M. (1999): {\it  Incompleteness of Markets 
driven by Mixed  Diffusion}, Finance  Stochastics, 4, 209-222.

\bibitem{bib:Black-Scholes} Black, F. and Scholes M. (1973): {\it
    The Pricing of Options and Corporate Liabilities}, J. Political Econ., 81, 637-659.



\bibitem{bib:Bouleau-Hirsch} Bouleau, N. and Hirsch, F. (1991):
  {\it Dirichlet Forms and Analysis on Wiener space}, De Gruyter, Berlin.

  \bibitem{bib:Bouleau-erreur} Bouleau, N. (2003): {\it Error
      Calculus for Finance and Physics}, De Gruyter, Berlin.

\bibitem{bib:Bouleau-erreur2} Bouleau, N. (2003): {\it Error Calculus and Path
Sensivity in Financial Models }, Mathematical Finance, 13-1, 
115-134.

\bibitem{bib:Bouleau-Chorro} Bouleau, N. and Chorro, Ch. (2004): {\it Structures
d'Erreur et Estimation paramétrique }, C.R. Accad. Sci. ,1338, 
305-311.




\bibitem{bib:Cox} Cox, J. and Ross, S. (1976): {\it  The Valuation of Options for Alternative 
Stochastics Processes}, J. of Financial Economics 3,  145-166.

\bibitem{bib:Cont} Cont R. and Tankov, P. (2004): {\it Financial Modelling with  Jump Processes}, Chapman \& Hall, London.

\bibitem{bib:Denis-1} Denis, L. and Martini, C. (2006): {\it A Theoretical Framework for the Pricing
of Contingent Claims in the Presence of Model Uncertainty}, Annals of
Applied Probability 16, 2, 827–852.

\bibitem{bib:Denis-2} Denis, L. and Kervarec M. (2009): {\it Utility Functions and Optimal investment in Non-Dominated Models}, preprint Universit\'{e} d'Evry.

  \bibitem{bib:Dupire} Dupire, B. (1994): {\it Pricing with Smile},
    RISK, January 1994.

\bibitem{bib:ElKaroui} El Karoui, N. and Rouge, R. (2000): {\it Pricing via Utility Maximisation 
and Entropy}, Math. Finance, 10, 259-276.

\bibitem{bib:Follmer} Follmer, H. and Schweizer, M. (1991): {\it Hedging of Contingent Claims under Incomplete Information}, in Davis, Elliott eds.: Applied Stochastic Analysis, 5, Gordon and Breach.

\bibitem{bib:Fouque} Fouque, J. P., Papanicolau, G. and Sircar, R. (2000). {\it Derivatives
in Financial Markets with Stochastic Volatility}, Cambridge
University Press, Cambridge.

\bibitem{bib:Fukushima} Fukushima, M.; Oshima, Y. and Takeda,
  M. (1994): {\it Dirichlet Forms and Markov Process}, De Gruyter, Berlin.



\bibitem{bib:Hagan} Hagan, P.; Kumar, D.; Lesniewsky, A. and Woodward, D.
(2002) {\it Managing Smile Risk} Willmot magazine 1,  84-108.

\bibitem{bib:Heston} Heston, S. (1993): {\it A Closed Form Solution for Options 
with Stochastic Volatility with Applications to Bond and Currency Options}, Review of Fin. 
Studies 6,  327-343.  

\bibitem{bib:Hull} Hull, J. and White, A. (1987): {\it  The Pricing of Options on assets with Stochastic
Volatilities}, J. of Financial and Quantitative Analysis 3,  281-300.

\bibitem{bib:Jeanblanc} Jeanblanc, M.; Yor, M. and Chesney M. (2009): {\it  Mathematical Methods for Financial Markets},
Springer Finance.

\bibitem{bib:Karatzas} Karatzas, I.; Lehoczky, J.; Shreve, S. and Xu, G. (1991): {\it Martingale and Duality Methods for Utility Maximisation in an Incomplete Market}, SIAM J. Control and Optimisation, 29(3), 702-730. 

\bibitem{bib:Kramkov} Kramkov, D. (1996): {\it Optional Decomposition of Supermartingales and Hedging Contingent Claims in Incomplete Security Market}, Prob. Theor. Relat. Fields, 105, 459-479.

\bibitem{bib:Lamb-Lap} Lamberton, D. and Lapeyre, B. (1995): {\it Introduction to Stochastic 
Calculus Applied to Finance}, Chapman \& Hall, London.

\bibitem{bib:Lyons} Lyons, T. (1995): {\it  Uncertain Volatility and the Risk Free Synthesis of Derivatives}, Applied Mathematical Finance 2,  117-133.


\bibitem{bib:Protter} Protter, Ph. (2004): {\it Stochastic Integration and Differential Equations: 
a new Approach}, Springer-Verlag, Berlin.

\bibitem{bib:Soner} Soner, H.M.; Touzi, N. and Zhan, J. (2010): {\it Quasi-sure Stochastic Analysis through Aggreagation}, 	arXiv:1003.4431v1




\end{thebibliography}
\end{document}